\documentclass[12pt,a4paper,english,reqno]{amsart}
\usepackage[T1]{fontenc}
\usepackage[latin9]{inputenc}
\usepackage{amsthm}
\usepackage{amssymb}
\usepackage{tikz}
\usepackage{tikz-cd}

\usetikzlibrary{arrows,chains,matrix,positioning,scopes}

\makeatletter
\tikzset{join/.code=\tikzset{after node path={%
\ifx\tikzchainprevious\pgfutil@empty\else(\tikzchainprevious)%
edge[every join]#1(\tikzchaincurrent)\fi}}}
\makeatother

\tikzset{>=stealth',every on chain/.append style={join},
         every join/.style={->}}

\makeatletter

\pdfpageheight\paperheight
\pdfpagewidth\paperwidth

\numberwithin{equation}{section}
\numberwithin{figure}{section}

\usepackage{amsfonts}

\usepackage{amsthm}     
\usepackage{amscd}      
\usepackage{euscript}
\usepackage[matrix,arrow,curve]{xy}
\usepackage{pb-diagram}
\usepackage{enumitem}

{
       \newtheorem{theorem}{Theorem}[section]
       \newtheorem{proposition}[theorem]{Proposition}
       \newtheorem{lemma}[theorem]{Lemma}

{\theoremstyle{definition}

       \newtheorem{definition}[theorem]{Definition}
       
}

\newcommand{\W}[2]{{H^{#1,#2}}}

\newcommand{\p}{\partial}

\newcommand{\Hol}{\mathcal{H}_{r,k}}

\newcommand{\pk}{\mathbb{CP}^{k}}

\newcommand{\nuk}{\nu_{k+1}(s)}
\newcommand{\nukzero}{\nu_{k+1,0}(s)}

\makeatother

\usepackage{babel}

\begin{document}
\title[The Gromov Limit for Vortex Moduli Spaces]{The Gromov Limit for Vortex Moduli Spaces}
\author{Gabriele La Nave}\address{University of Illinois, Urbana-Champaign, IL, USA\\lanave@illinois.edu } \author{Chih-Chung Liu}\address{Department of Mathematics, National Cheng-Kung University, Tainan, Taiwan\\cliu@mail.ncku.edu.tw}
\maketitle

\begin{abstract}
We generalize the descriptions of vortex moduli spaces in \cite{Br} to more than one section with adiabatic constant $s$. The moduli space is topologically independent of $s$ but is not compact with respect to $C^\infty$ topology. Following \cite{PW}, we construct a Gromov limit for vortices of fixed energy, an attempt to compactify the moduli space.
\end{abstract}

\section{Introduction}
The study of vortex equations finds its origin in Ginzburg-Landau's descriptions of the field configurations of superconducting material (cf. \cite{JT}). There, the energy functional is given in the form of Yang-Mills-Higg functional, which depends on the electromagnetic potential $D$ and the wave function $\phi$ of Cooper pairs of electrons. Stable configurations are governed by minimizing the enregy functional, and the minimizing equations are known as the vortex equations.

The theory is mathematically modelled by equations on Hermitian vector bundles $(E,H)$ of degree $r$ over closed K\"ahler manifolds $M$, where the variables consist of $D \in \mathcal{A}(H)$, an $H$-unitary connection, and global smooth section $\phi$. The K\"ahler form $\omega$ of $M$ is normalized so that $Vol_\omega(M)=1$. The Yang-Mills-Higgs energy functional contains two extra terms from classical Yang-Mills functional, arisen from sections:

\begin{equation}
YMH_{1,1}(D,\phi):=\left|\left|F_{D}\right|\right|_{L^{2}}^{2}+\left|\left|D\phi\right|\right|_{L^{2}}^{2}+\frac{1}{4}\left|\left|\phi\otimes\phi^{*_H}-\tau \right|\right|_{L^{2}}^{2},\label{YMH classical}
\end{equation}

\noindent where $F_D$ is the curvature form of connection $D$ and $\tau$ is a real parameter. The minimizing equations can be deduced by Bogomol'yi arguments (\cite{Br}):

\begin{equation}
\begin{cases}
F_{D}^{(0,2)}=0\\
D^{(0,1)}\phi=0\\
\sqrt{-1}\Lambda F_{D}+\frac{1}{2}(\phi\otimes\phi^{*_H}-\tau)=0.
\end{cases}\label{s-vortex classical}
\end{equation}

\noindent Namely, among pairs $(D,\phi)$ such that $D$ is integrable and $\phi$ is $D$-holomorphic, the last equation imposes a relation on mean curvature and the norm of the sections. As a standard principle in gauge theoretic equations, the existence of solutions is equivalent to a $\phi$ and $\tau$ dependent stability of the line bundle $E$ (cf. \cite{Br},\cite{Br1}).  For line bundles $E=L$, the absence of proper subsheaves turns the stability condition dependent only on the  parameter $\tau$. In \cite{Br} and \cite{Br1}, it is proved that the necessary condition for vortex to exist, followed by integrating the third equation of \eqref{s-vortex classical},

\begin{equation}
\tau \geq 4 \pi r
\label{stability condition classical}
\end{equation}

\noindent is also sufficient. Solutions to \eqref{s-vortex classical} are clearly invariant under standard unitary gauge actions. Within the stable range, the moduli space of solutions to \eqref{s-vortex classical}, or the gauge classes of \emph{vortices}, has been explicitly described in \cite{Br}. The moduli space of vortices to \eqref{s-vortex classical} is precisely $Div_+^r M$, the space of degree $r$ effective divisors on $M$, and is topologically independent of $\tau$.

Generalizations of the classical results for the case of line bundles have been made in \cite{Ba} and \cite{L}. We consider Yang-Mills-Higgs functional defined on $k+1$ sections with a scaling parameter $s$. The parameter $\tau$ can be absorbed into $s$ (cf. \cite{L}) and we rewrite

\begin{equation}
YMH_{k+1,s}(D,\phi):=\frac{1}{s^{2}}\left|\left|F_{D}\right|\right|_{L^{2}}^{2}+\sum_{i=0}^k\left|\left|D\phi_i\right|\right|_{L^{2}}^{2}+\frac{s^{2}}{4}\left|\left|\sum_{i=0}^k|\phi_i|^2_H-1 \right|\right|_{L^{2}}^{2},\label{YMH}
\end{equation}

\noindent where $\phi$ is an abbreviation for $(\phi_i)_{i=0}^k$. The corresponding vortex equations are

\begin{equation}
\begin{cases}
F_{D}^{(0,2)}=0\\
D^{(0,1)}\phi_i=0\;\;\forall i\\
\sqrt{-1}\Lambda F_{D}+\frac{s^{2}}{2}(\sum_{i=0}^k|\phi_i|^2_H-1)=0.
\end{cases}\label{s-vortex}
\end{equation}

\noindent The stability condition is then

\begin{equation}
s^2 \geq 4\pi r.
\label{stability condition}
\end{equation}

\noindent Solutions to \eqref{s-vortex} are again invariant under unitary gauge group $\mathcal{G}$. We then define, for $s$ in the stable range, the gauge class of solutions

\begin{definition}
\[\nu_{k+1}(s):=\{(D,\phi)\in\mathcal{A}(H)\times\Omega^0(L)\times\cdots\times\Omega^0(L)\;|\;\eqref{s-vortex} \; hold\}/\mathcal{G}.\]
\label{definition of s vortex moduli spaces}
\end{definition}

\noindent These spaces are topologically independent of $s$ and fibers over a space with explicit description. Detailed descriptions are provided in section 3. Each class $[D_s,\phi_s]\in\nuk$ represents a unique holomorphic structure for the line bundle $L$. Nevertheless, the topological structures of the line bundle $L$, determined by the mean value of the corresponding curvatures $F_{D_s}$, remain undisturbed until the adiabatic limit $s=\infty$. As noted in \cite{Ba}, the formal limit of \eqref{s-vortex} as $s\to\infty$ is:

\begin{equation}
\begin{cases}
F_{D}^{(0,2)}=0\\
D^{(0,1)}\phi=0\\
\sum_{i=0}^k|\phi_i|^2_H-1=0,
\end{cases}\label{infinity-vortex}
\end{equation}

\noindent signaling some topological distinctions from that of \eqref{s-vortex}. In particular, for the case of one section $k=0$ the third equation above requires the global section to be non-vanishing, which only exists on trivial line bundles. Moreover, the norms of the sections are no longer constrained by curvature. These inconsistencies signal bubbling phenomenon of vortex moduli spaces along some variation of vortices. The main theme of this paper is to study such situations and provide explicit description of the bubble formation. Moreover, we construct a reasonable limiting object for $[D_s,\phi_s]$ without topological loss.

The bubbling phenomenon depends crucially on the dynamics of solutions to \eqref{s-vortex} and occurs when singularities form in the limiting solution. As will be shown in section 5, singularities of limiting solutions are due to accumulations of common zeros, or base points, of $\phi_i$'s at the boundary. In particular, consider the generic open subset

\begin{definition}
\[\nukzero :=\{[D,\phi_0,\ldots,\phi_k]\in\nuk\;|\;\cap_i \phi_i^{-1}(0) = \emptyset\}.\]
\label{vortices without common zero}
\end{definition}

\noindent Analytic results from \cite{L} show that no bubbling phenomenon occurs when the entire convergence takes place in $\nukzero$. We are therefore mainly interested in families of vortices for which new common zeros form at infinity. That is, when the sequence in $\nukzero$ converges to a boundary point. The dense open subset above may be identified diffeomorphically by a family of holomorphic maps from $M$ to $\pk$ (cf \cite{Gri}). The Dirichlet energies of these associated maps are precisely the degree of the bundle. For closed Riemann surfaces $M=\Sigma$, these are holomorphic curves with bounded energies and we may then apply results from \cite{Gro} and \cite{PW} to establish their convergence behaviours of Gromov type. The limiting behaviours and objects known as the "bubble trees" are compatible with vortex moduli spaces. The descriptions of bubble trees require some amount of work. For a fixed line bundle $L$ over a closed Riemann surfaces $\Sigma$, we first prove the following theorem.

\begin{theorem}[Formal Removal of Singularities]
Fix a Hermitian line bundle $(L,H)$ over $\Sigma$. Given a sequence of vortices $\{[D_s,\phi_s]\}\subset \nukzero$ approaching the boundary of $\nuk$, there exists a finite set of points $\{p_1,\ldots,p_N\}\subset \Sigma$, integers $\{a_1,\ldots,a_N\}\subset \mathbb{N}$ such that $\sum_ja_j \leq r$, and vortices $[D_s^\prime,\phi_s^\prime]$ with smooth (subsequential) limit $[D_0,\phi_0]$ on line bundle

\[L_0:=L \otimes_j \mathcal{O}(-a_j p_j),\]

\noindent such that

\begin{itemize}
\item $[D_s^\prime,\phi_s^\prime] = [D_s,\phi_s]$ on $\Sigma \backslash \{p_1,\ldots,p_N\}$ (via the isomorphism $L_0 \simeq L$ on $\Sigma \backslash \{p_1,\ldots,p_N\}$).
\item $D_s^\prime$ and $\phi_s^\prime$ satisfy the vortex equation
\begin{equation}
\begin{cases}
  D_s^{\prime (0,1)} \phi_{s,i}^\prime = 0 \;\;\forall i \\
  \sqrt{-1}\Lambda F_{D_s^\prime} + \frac{s^2}{2}\left(\sum_{i=0}^k|\phi_{s,i}^\prime|^2_H-1\right)=0.   \\
  \end{cases}
  \label{s vortex prime intro}
\end{equation}
on $L_0 \to \Sigma$.
\item $[D_0,\phi_0]$ satisfies
\begin{equation}
\begin{cases}
  D_0^{(0,1)} \phi_{0,i} = 0 \;\;\forall i \\
  \sum_{i=0}^k|\phi_{0,i}|^2_H-1=0.   \\
  \end{cases}
\label{infity vortex after removal of singularity intro}
\end{equation}
on $L_0 \to \Sigma$.
\end{itemize}
\label{Removal of Singularities intro}

\end{theorem}

The extended line bundle $L_0$ is of degree $r - \sum_j a_j$. The reduction of degree suggests concentration of energy of vortices near singularities. Appropriately rescaling nearby coordinates by some factor $t_j(s)$, we may smooth out the energy density and define vortices on $\mathbb{C}$, which is identified by $\mathbb{S}^2 \backslash \{p^+\}$, where $p^+$ is the north pole, via stereographic projection. The limiting objects are determined by the rates of energy blow ups.

\begin{theorem}[Renormalization]
For each $p_j$ in Theorem \ref{Removal of Singularities}, there exists $\epsilon > 0$ so that the geodesic disc $B(p_j,\epsilon)$ is conformally equivalent to $B_s \subset \mathbb{S}^2$, an increasing family of domains with $\cup_s B_s = \mathbb{S}^2 \backslash \{p^+\}$, and  the followings hold:

\begin{itemize}
\item The pullbacked vortices $[D_s^*,\phi_s^*]$ on $B_s$, satisfying
\begin{equation}
\begin{cases}
  D_s^{*0,1} \phi_{s,i}^* = 0 \\
  \sqrt{-1}\Lambda_s^* F_{D_s^*}+\frac{s^2}{2t_j(s)^2}(\sum_{i=0}^k|\phi_{s,i}^*|^2_H-1)=0\\
  \end{cases}
  \label{pullbacked s vortex intro}
\end{equation}
\noindent on pullbacked line bundle $L_s$ over $B_s$, coincide with the vortices defined by pullbacked holomorphic functions $\tilde{f}_s: B_s \to \pk$.

\item Exactly one of the followings holds true:

\begin{enumerate}[label=(\alph*)]
\item There exists a $C^1_{loc}$-convergent subsequence of $\{[D_s^*,\phi_s^*]\}$ whose limit $[D_j,\phi_j]$ satisfies

\begin{equation}
\begin{cases}
  D_j \phi_{j,i} = 0 \\
 \sum_{i=0}^k|\phi_{j,i}|^2_H-1=0\\
  \end{cases}
  \label{pullbacked infinity vortex holomorphic sphere intro}
\end{equation}

defined on the entire $\mathbb{S}^2$. That is, a holomorphic sphere in $\pk$ bubbles off.

\item There exsits points $\{p_j^1,\ldots,p_j^{N_j}\}\subset \mathbb{S}^2$, nonnegative integers $a_j^0,a_j^1,\ldots,a_j^{N_j}$, and a $C^1_{loc}$-convergent subsequence of $\{[D_s^*,\phi_s^*]\}$ on $\mathbb{S}^2 \backslash \{p_j^1,\ldots,p_j^{N_j},p^+\}$, whose limit $[D_j,\phi_j]$ satisfies
\begin{equation}
\begin{cases}
  D_j \phi_{j,i} = 0 \\
  \sqrt{-1}\Lambda^* F_{D_j}+\frac{1}{2}(\sum_{i=0}^k|\phi_{j,i}|^2_H-1)=0\\
  \end{cases}
  \label{pullbacked infinity vortex intro}
\end{equation}
\noindent on a degree $a_j^0$ line bundle $L_j$ over $\mathbb{S}^2 \backslash \{p_j^1,\ldots,p_j^{N_j},p^+\}$. Moreover, $(L_j,D_j,\phi_j)$ is the $C^1$ limit of $(L_s^*,D_s^*,\phi_s^*)$.

\end{enumerate}

\item On $\mathbb{S}^2$, $\sqrt{-1}\Lambda^*F_{D_j}$ is a distribution given by a smooth function plus $\sum_{l=1}^{N_j} a_j^l \delta (p_j^l)$.

\end{itemize}
\label{renormalization intro}
\end{theorem}

\noindent The principle to obtain the results above is to associate each generic vortex with a holomorphic map from $\Sigma$ to $\pk$ (see section 3) and apply analytic results from \cite{PW}. The limiting vortex $[D_j,\phi_j]$ above is given by a limiting map $\tilde{f}_{p_j}$ as well. Standard Morrey estimate and bootstrapping arguments allow one to extend $\tilde{f}_{p_j}$ holomorphically to the entire $\mathbb{S}^2$. Analogous extension is possible for vortices but requires certain adjustments.

\begin{theorem}[Removal of Singularities for Limiting Vortices]
Continuing with the setting of Theorem \ref{Removal of Singularities intro} and Theorem \ref{renormalization intro}, the conformal transformations may be modified so that the limiting vortex $[D_j,\phi_j]$ may be extended, after an appropriate gauge transformation, across $p^+$. The extended pair gives rise to a vortex defined on a nontrivial line bundle $L_j$ over $\mathbb{S}^2$ with degree $\leq a_j$. The metric on $\mathbb{S}^2$ may come with a conic singularity.
\label{extension of vortices intro}
\end{theorem}

\noindent To this end we have associated a "bubble", meant to smoothen the energy spike, to each point where energy density blows up (the "bubble point"). The description is not yet satisfactory because the inequality at the end of the theorem above may be strict. That is, the bubble does not necessarily retain all the energy near $p_j$. Following \cite{PW}, we may modify the renormalization process to accurately account for all the energy concentrated near each $p_j$ by a finite sequence of bubbles. At the end, we obtain a bubble tree $T$, a wedge sum of $\Sigma$ and $\mathbb{S}^2$'s. The vortices $[D_s,\phi_s]$ "Gromov converges" to a vortex $[\mathcal{D},\varPhi]$ on $T$. The vague terms here will be precisely defined in the due course.

\begin{theorem}[Bubble Tree]
The vortices $V_s:=\{[D_s,\phi_s]\}$ on a degree $r$ line bundle $L$ over $\Sigma$ Gromov converge to a vortex $V:=[\mathcal{D},\varPhi]$ over a degree $r$ line bundle $\mathcal{L}$ over a bubble tree $T$ defined by

\begin{equation}
  T:= T_0 \vee T_1 \vee \cdots  \vee T_{N_V},
  \label{bubble tree wedges}
\end{equation}

\noindent where $T_0 = \Sigma$. For $n\geq 1$, each $T_n$ is a disjoint union of 2-spheres with either round metric $g_{\mathbb{S}^2}$ or conic metric $g_\beta$.
\label{bubble tree intro}
\end{theorem}

\section{Background and Established Results}
This section briefly summarizes results from \cite{Br} and \cite{L} on the geometric descriptions of some special $\nuk$'s. Analytic techniques developed in \cite{Br} assume $k,s=1$, but they are by no means special to this particular values. We will therefore cite those results with general $k\in\mathbb{N}$ and $s$ in the stable range. Readers familiar with these work may skip to the next section.

Let $(M,\omega)$ be a closed K\"ahler manifold of unit volume. $(L,H)$ is a Hermitian line bundle of degree $r$ over it. Denote by $\mathcal{A}(H)$ the space of $H$-unitary connections and $\Omega^0(L)$ the space of smooth global sections. The symmetries of $(L,H)$ considered are denoted by $\mathcal{G}_{\mathbb{C}}$ and $\mathcal{G}$ , called the complex gauge group and unitary gauge groups, respectively. Both groups act on $\mathcal{A}(H)$, $\Omega^0(L)$, the metrics, complex structures, and curvature forms of $L$ in the standard ways (see \cite{Br} or \cite{K}). The necessity of the stability condition $s^2 \geq 4\pi r$ implies immediately that $\nuk = \emptyset$ for all $s$ with $s^2 < 4\pi r$. For the critical value $s^2=4\pi r$, the third equation of \eqref{s-vortex} requires that all sections $\phi_i$ to be trivial, and therefore $\nuk$ is precisely $\mathcal{A}(H) / \mathcal{G}_{\mathbb{C}}$, or the space of holomorphic structures of $L$. For $M=\Sigma$, the space corresponds to the Jacobian torus of degree $r$, $Jac^r \Sigma$.

Analytic discussion enters when $s^2 > 4\pi r$. Similar to the search of Hermitian-Einstein connections, solving the tensorial vortex equation modulo unitary gauge group is equivalent to searching for special Hermitian metric modulo complex gauge group. The restatement of the problem by variation of metrics invites classical analytic tools from \cite{kw} to enter the central argument.

We briefly summarize the correspondence of the two aspects. For a Hermitian line bundle $(L,H)$, it is a classical fact that the space of unitary connection $\mathcal{A}(H)$ and the space $\mathcal{C}$ of holomorphic structures, or the collection of $\mathbb{C}$-linear operators

\[\bar{\p}_L:\Omega^0(L) \to \Omega^{0,1}(L)\]

\noindent satisfying Leibiniz rule and $\bar{\p}_L \circ \bar{\p}_L=0$, identify each other. Fix $k+1$ global sections $\phi=(\phi_i)_{i=0}^k$. The original tensorial approach to solve \eqref{s-vortex} is then equivalent to finding a holomorphic structure $\bar{\p}_L$, that makes all $\phi_i$ holomorphic, so that the corresponding unitary connection $D$ and curvature $F_D$ satisfy the equation. This approach, however, is rather abstract. The alternative, or the scalar approach, picks an arbitrary pair in the space of \emph{holomorphic pair}

\begin{equation}
\mathcal{N}_{k+1} := \{(\bar{\p}_L,\phi)\in \mathcal{C} \times \Omega^0(L) \times \cdots \times \Omega^0(L)\;|\;\bar{\p}_L(\phi_i)=0 \;\forall i\}.
\label{holomorphic pair}
\end{equation}

\noindent  We then look for a special $H_s \in \mathcal{H}$, the space of Hermitian structure, whose corresponding connection, and therefore curvature form $F_{D_s}$, together with the given sections satisfy the third equations of \eqref{s-vortex}:

\begin{equation}
\sqrt{-1}\Lambda F_{D_s} + \frac{s^2}{2}(\sum_{i=0}^k|\phi_i|^2_{H_s}-1)=0.
\label{scalar s-vortex}
\end{equation}

\noindent For a line bundle $L$, the complex gauge group $\mathcal{G}_{\mathbb{C}}$ acts transitively on $\mathcal{H}$. In particular, the special metric $H_s$ and the background metric $H$ are related by

\[H_s = e^{2u_s} H,\]

\noindent where $u_s$ is a real smooth function on $M$. The corresponding curvature $F_{D_s}$ is then related to the background curvature $F_H$ by

\begin{equation}
\sqrt{-1}\Lambda F_{D_s}=\sqrt{-1}\Lambda F_H - \Delta u_s,
\label{variation of curvatures}
\end{equation}

\noindent where $\Delta$ is the positive definite Laplacian determined by the K\"ahler form $\omega$. Let $c_1 = 2\pi r$, and $c(s)=2c_1 - \frac{s^2}{2}$, which is negative for $s$ in the stable range. Also let $\psi$ be the unique solution to the Poisson equation

\begin{equation}
\Delta \psi = \sqrt{-1} \Lambda F_H -c_1.
\label{equation for psi}
\end{equation}

\noindent It can be readily verified that solving \eqref{scalar s-vortex} above is equivalent to solving the following {\em Kazdan-Warner equation}

\begin{equation}
\Delta \varphi_s+\frac{s^2}{2}h e^{\varphi_s} - c(s)=0,
\label{Kazdan-Warner equation}
\end{equation}

\noindent where $\varphi_s=2(u_s-\psi)$ and the norm function

\begin{equation}
h=-e^{2\psi}\sum_{i=0}^k |\phi_i|_H^2
\label{norm function}
\end{equation}

\noindent is non-positive and vanishes precisely at the common zeros of $\phi_i$'s, an effective divisor $\mathcal{E}$ of degree $\leq r$. For these choices of $c(s)$ and $h$, techniques developed in \cite{kw} guarantee a unique smooth solution $\varphi_s$ for each finite $s$ in the stable range. The unique existence is proved by the standard arguments of upper and lower solutions of elliptic operators and applications of maximum principles.

The analytic results imply that given a holomorphic pair $(\bar{\p}_L,\phi)\in\mathcal{N}_{k+1}$ in \eqref{holomorphic pair}, the special metric $H_s$ to solve \eqref{scalar s-vortex} is uniquely determined if $s^2 > 4\pi r$. One may readily recognizes the gauge ambiguities and conclude that the space of $\mathcal{G}$-classes of solutions to \eqref{s-vortex} corresponds bijectively to $\mathcal{N}_{k+1} / \mathcal{G}_{\mathbb{C}}$. For $k=0$, this space is further identified, up to a $\mathcal{G}_{\mathbb{C}}$ action, with the space $Div_+^r (M)$ of effective divisor of degree $r$ (cf. \cite{Br}). Indeed, an effective divisor of degree $r$ determines a holomorphic line bundle $L=\mathcal{O}(\mathcal{E})$ with holomorphic structure $\bar{\p}$, and all global sections vanishing along $\mathcal{E}$ are in one complex gauge orbit. We have,

\begin{theorem}[Description of $\nu_1(s)$ for $M$]

\[\nu_1(s)=
\begin{cases}
\emptyset &; s^2 <  \pi r \\
\mathcal{A}(H)/\mathcal{G}_{\mathbb{C}} &; s^2 = 4 \pi r \\
Div_+^r(M)&; s^2 > 4 \pi r. \\
\end{cases} \]
\label{moduli spaces for one section manifold}
\end{theorem}

\noindent For a closed Riemann surface $M=\Sigma$, $Div_+^r(\Sigma)$ is precisely the space of unordered $r$ points, or the symmetric space and $\mathcal{A}(H) / \mathcal{G}_{\mathbb{C}}$ is identified with the Jacobian torus of degree $r$. We have

\begin{theorem}[Description of $\nu_1(s)$ for $\Sigma$]

\[\nu_1(s)=
\begin{cases}
\emptyset &; s^2 <  \pi r \\
Jac^r \Sigma &; s^2 = 4 \pi r \\
Sym^r \Sigma&; s^2 > 4 \pi r. \\
\end{cases} \]
\label{moduli spaces for one section}
\end{theorem}

\noindent $\nuk$ for general $k$ has been described in \cite{BDW} for the case $M=\Sigma$. In the next section, we provide a general description which can be easily specialized to the case of Riemann surfaces. Since we are interested in the adiabatic limit $s \to \infty$, we will from now on assume $s^2 > 4\pi r$, ruling out the first two possibilities in Theorem \ref{moduli spaces for one section manifold} and Theorem \ref{moduli spaces for one section}.

\section{Generalized Vortex Moduli Spaces and Maps to Projective Spaces}

We provide the general descriptions for $\nuk$ here. In the space $\mathcal{N}_1$ in \eqref{holomorphic pair}, we see that after a complex structure $\bar{\p}_L$ is fixed, holomorphic sections are determined by effective divisors of degree $r$ up to $\mathcal{G}_\mathbb{C}$ gauge. For holomorphic tuples with $k+1$ sections, it is natural to analogously identify each vortex by the tuple of $k+1$ divisors defined by each section. However, ambiguities and restrictions arise. An immediate restriction is that all divisors must define isomorphic holomorphic line bundle, as we are fixing one holomorphic structure at a time. In another words, all divisors must be linearly equivalent. Therefore we start with the space

\begin{definition}
\[\mathbb{E}_{k+1}:=\{(E_0,\ldots,E_k) \in (Div_r^+(M))^{\times k+1}\;|\;E_0 \sim \cdots \sim E_k\}.\]
\label{divisor labelling space}
\end{definition}

\noindent This is a closed subset of $(Div_r^+(M))^{\times k+1}$ and therefore compact. It however is not an identification of $\nuk$. An effective divisor $E_i$ determines a global section $\phi_i$ up to an element in $\mathcal{G}_\mathbb{C}$, and the gauge ambiguity for each $i$ need not be unitary gauge equivalent. That is, $\mathbb{E}_{k+1}$ only identifies $\mathcal{N}_{k+1}$ up to a $(\mathcal{G}_\mathbb{C})^{k+1}$ action, which is larger than the diagonal action $\mathcal{G}_\mathbb{C}$ on $\mathcal{N}_{k+1}$ used to define $\nuk$. Therefore, each element in $\mathbb{E}_{k+1}$ determines a vortex element in $\nuk$ up to a $(\mathcal{G}_\mathbb{C})^{k+1} / \mathcal{G}_\mathbb{C}$ orbit. More precisely, we have

\begin{theorem}
The space $\nuk$ fibers over the space $\mathbb{E}_{k+1}$ with toric fiber $(\mathbb{C}^*)^{k+1} / \mathbb{C}^*$.
\label{generalized moduli space}
\end{theorem}

\begin{proof}
For each $p \in M$, let $U_p \subset  \mathbb{E}_{k+1}$ be the open subset such that no divisor contains $p$. Since $\mathbb{E}_{k+1}$ is compact, it may be covered by finitely many such open subsets $U_1,\ldots,U_M$, each equipped with a based point  $\beta_j$ away from all divisors in $U_j$. Each $(E_0,\ldots,E_k) \in U_j$ determines a holomorphic structure $\bar{\p}_L$ of $L$ (ie $L\simeq\mathcal{O}(E_0)\simeq\cdots\simeq\mathcal{O}(E_k)$). Every $E_i$  determines a global holomorphic section $\phi_i$ up to a nonzero constant. Therefore, two such sections are equal if and only if their values are equal at the based point $\beta_j$. Explicitly, consider the projection

\[\tilde{\pi}: \mathcal{N}_{k+1} \to \mathbb{E}_{k+1}\]

\noindent defined by
\[\tilde{\pi}(\bar{\p}_L,\phi_0,\ldots,\phi_k)=(\phi_0^{-1}(0),\ldots,\phi_k^{-1}(0)).\]

\noindent This is a fiber bundle with fiber $(\mathbb{C}^*)^{k+1}$, where the local trivialization over $U_j$ is given by

\[\tilde{\rho}_j(\bar{\p}_L,\phi_0,\ldots,\phi_k) = (\phi_0^{-1}(0),\ldots,\phi_k^{-1}(0),\phi_0(\beta_j),\ldots,\phi_k(\beta_j)).\]

\noindent Finally we note that in the identification $\nuk \simeq \mathcal{N}_{k+1}/\mathcal{G}_\mathbb{C}$, the gauge action does not affect the zeros of the sections and therefore the projection $\tilde{\pi}:\mathcal{N}_{k+1}\to \mathbb{E}_{k+1}$ descends to the quotient $\pi: \nuk \to \mathbb{E}_{k+1}$. It follows that $\nuk$ is a bundle over $\mathbb{E}_{k+1}$ with fiber $(\mathbb{C}^*)^{k+1} / \mathbb{C}^*$.

\end{proof}

\noindent The theorem is clearly consistent with \cite{Br}, where $k=0$.

To relate $\nuk$ to the space of holomorphic maps, we consider its open dense subset:

\begin{definition}
\[\nukzero:= \{[D,\phi_0,\ldots,\phi_k]\in\nuk\;|\;\cap_i \phi_i^{-1}(0)=\emptyset\}.\]
\label{definition of nonvanishing vortices}
\end{definition}

\noindent On this subset, we may control the special gauges (i.e. solutions to \eqref{Kazdan-Warner equation}) and guarantee the existence of their smooth limit.  This is due to the fact that the norm function $h$ in \eqref{norm function} is strictly negative, and we may uniformly bound super and sub solutions for the elliptic Kazdan-Warner equations \eqref{Kazdan-Warner equation}. The general analytic statement from \cite{L} is

\begin{theorem} [Asymptotic Behaviours on $\nukzero$]

On a compact Riemannian manifold $M$ without boundary, let $c_1$ be any constant, $c_2$ any positive constant,
and $h$ any negative smooth function. Let $c(s) = c_1 -c_2 s^2$, for each $s$ large enough,
the unique solutions $\varphi_s \in C^\infty(M)$ for the equations

\[
\Delta\varphi_{s}=c(s)-s^2he^{\varphi_s}
\]

\noindent are uniformly bounded in $\W{l}{p}$ for all $l \in \mathbb{N}$ and $p \in [1,\infty]$. Moreover, in the limit
$s \to \infty$, $\varphi_{s}$ converges smoothly (i.e. uniformly in all $H^{l,p}$) to

\[\varphi_\infty = \log\left(\frac{c_2}{-h}\right),\]

\noindent the unique solution to

\[
he^{\varphi_{\infty}}+c_{2}=0.
\]

\label{Main Theorem L}
\end{theorem}

\noindent The theorem in particular rules out the formation of bubble point away from the boundary of $\nuk$.

Topologically, the space $\nukzero$ is identified with $\Hol$, the space of degree $r$ holomorphic maps from $M$ to $\pk$. For every $f \in \Hol$, consider the following background data: On the anti-tautological line bundle $\mathcal{O}(1)$ over $\pk$ and its pullbacked bundle $L = f^*\mathcal{O}(1)$ over $M$, let $\phi:=(\phi_i)_{i=0}^k$ be the global holomorphic (with respect to the pulled back holomorphic structure) sections on $L$ pullbacked from the hyperplane sections $z_0,\ldots,z_k$ on $\mathcal{O}(1)$ via $f$. $\mathcal{O}(1)$ is equipped naturally with the Fubini-Study metric, which is also pulled back to be the background metric $H$ on $L$, and therefore determines an unitary connection $D$. We then gauge transform the initial data $[D,\phi]$ with complex gauge determined by Kazdan-Warner equation \eqref{Kazdan-Warner equation} into $[D_s,\phi_s]$ that solves vortex equations. These sections clearly have no common zeros, and we define

\begin{equation}
\Phi_s : \Hol \to \nukzero
\label{correspondence maps vortices}
\end{equation}

\noindent by

\[\Phi_s(f) = [D_s,\phi_s].\]

\noindent The correspondence is in fact a diffeomorphism with a natural inverse

\[\Phi_s^{-1} : \nukzero \to \Hol\]

\noindent given by

\[\Phi_s^{-1}([D_s,\phi_s])(p) := f_s (p) = [\phi_{s,0}(p):\cdots:\phi_{s,k}(p)].\]

\noindent As a smooth map, $f_s$ is clearly well defined as sections do not vanish simultaneously. Also, a different choice of trivilization amounts to multiplication of all components with a nonzero constant and therefore does not alter the definition. Moreover, each $f_s$ is holomorphic with respect to the complex structure of $\Sigma$ given by $D_s$, which corresponds to conformal deformations of K\"ahler form $\omega$. Therefore the complex structure is independent of $s$ and the map $\Phi_s^{-1}$ is indeed well defined.  See for example, \cite{L} for detailed verifications.

In the context of Theorem \ref{generalized moduli space}, $\nukzero$ is the restriction of $\nuk$ over the generic open subset

\[\mathbb{E}_{k+1,0} := \{(E_0,\ldots,E_k)\in \mathbb{E}_{k+1}\;|\;\cap_i E_i = \emptyset\}.\]

\noindent For Riemann surface $M=\Sigma$, we have an explicit interpretation of Theorem \ref{generalized moduli space} in terms of holomorphic maps, introduced in \cite{BDW}. Indeed, for $M=\Sigma$, $Div_r^+(\Sigma)=Sym^r \Sigma$, the symmetric $r$ product. For an element $(E_0,\ldots,E_k) \in \mathbb{E}_{k+1}$, we consider the divisor formed by their intersection and a holomorphic map from $\Sigma$ to $\pk$ by the Theorem of Abel-Jacobi. Precisely, let $E=\cap_i E_i \in Sym^l \Sigma$, counting multiplicities, for $0\leq l\leq r$ and $E_i' = E_i -E$. For each $i \geq 1$, the divisor $E_i' - E_0'$ has value zero under Abel-Jacobi map and therefore determines a meromorphic function $\phi_i$ on $\Sigma$. These maps have no common zero and define a holomorphic map of degree $r-l$ locally given by
\[f_\phi(z):= [1,\phi_1(z),\ldots,\phi_k(z)].\]

\noindent The map is unique up to a choice of image of the base point, which corresponds to a choice of representative from the fiber $(\mathbb{C}^*)^{k+1}/\mathbb{C}^*$ in Theorem \ref{generalized moduli space}. The representative is determined by Kazdan-Warner equation. We therefore uniquely associate a vortex $[D,\phi] \in \nuk$ with an element $(E,f_\phi) \in Sym^l \times \mathcal{H}_{r-l,k}$. This is the well known identification of $\nuk$ with the \emph{Uhlenbeck Compactification} from \cite{BDW}.

\begin{theorem}[Uhlenbeck Compactification \cite{BDW}]
For all finite $s$ large enough, the space $\nuk$ is homeomorphic to the stratification

\begin{equation}
\overline{\mathcal{H}}_{r,k} := \bigsqcup_{l=0}^r \left(Sym^l \Sigma \times \mathcal{H}_{r-l,k}\right).
\end{equation}

\label{Uhlebceck compactification}
\end{theorem}

\noindent The topology of Uhlenbeck compactification is given sequentially.  $(E_i,f_i) \to (E,f)$ if an only if

\begin{itemize}
\item $f_i \to f$ in $\mathcal{C}_0^\infty(\Sigma - E) $ topology, and
\item $e(f_i) \to e(f)$ in weak* topology.
\end{itemize}

\noindent Here, $e(f)$ is the energy density $|df|^2$ of $f$ with respect to $\omega$ and Fubini-Study metric on $\pk$. The weak convergence above says that for all $g \in \mathcal{C}^\infty(\Sigma)$, we have

\[\int_\Sigma g e(f_i) := \int_\Sigma g f_i^* \omega_{FS} \to \int_\Sigma g f^*\omega_{FS} := \int_\Sigma g e(f)\]

\noindent as $i \to \infty$. The precise homeomorphic correspondence is exhibited therein.

\section{Gromov Compactness}

Having identified $\nukzero$ with $\Hol$, we intend to study convergence behaviours of vortices in terms of maps. We state relevant definitions and theorems from \cite{PW} regarding Gromov compactness of the space of $jJ-$ holomorphic maps in the form applicable to the aimed resultss. From this section on, we discuss the case $(M,\omega) = (\Sigma,\omega)$, a closed K\"ahler Riemann surface with complex (conformal) structure $j$ compatible with $\omega$. Also, fix a symplectic manifold $Z$ with almost complex structure $J$. When no confusion arises, we use the same notation for energy density function

\[f^* \omega_Z := e(f) \omega_\Sigma\]

\noindent and the measure it represents on $C^\infty(\Sigma)$. Proofs from \cite{PW} will be outlined or sketched only when they are relevant to our applications.

We need an important estimate on the energy density function by its integral, a standard result following from Bochner type estimate on $\Delta e(f)$. The inequality allows one to apply the Theorem of Arzela-Ascoli.

\begin{proposition} (Energy Estimate - Theorem 2.3 in \cite{PW})
There exists positive constants $C$ and $\epsilon_0$, depending only on the complex geometry of $\Sigma$, such for all $C^1$ J-holomorphic map $f : \Sigma \to Z$ and geodesic disc $B(2r)$ with radius $2r$, with the property that

\begin{equation}
E(2r):= \int_{B(2r)} e(f) \leq \epsilon_0,
\label{main energy bound}
\end{equation}

\noindent we have

\begin{equation}
\sup_{B(r)} e(f) \leq \frac{C}{r^2} E(2r).
\label{main energy estimate}
\end{equation}
\label{Bochner energy estimate}
\end{proposition}

The removability of singularity of vortices depends on the following basic theorem on extension $jJ$-holomorphic map across a punctured disc.

\begin{theorem}[Removable Singularity for Maps, Theorem 3.7 of \cite{PW}]
Let $(Z,J)$ be a complex manifold with complex structure $J$ and $(\Sigma,j)$ be a Riemann surface with complex structure $j$. Let $B \backslash \{p\}$ be a punctured disc of $\Sigma$ and

\[f: B \backslash \{p\} \to Z\]

\noindent be a $jJ-$holomorphic map of finite energy. Then $f$ extends to a $jJ-$ holomorphic map

\[\bar{f} : B \to Z.\]

\label{removable singularity for maps}
\end{theorem}

The extension and its regularity follow from Morrey's Lemma (Theorem 2.1 of \cite{Mo}), and the required energy estimate is achieved by an isoperimetric inequality in \cite{PW} (Proposition 3.4).  We will also need the fact that the energies of non-constant holomorphic maps can not be arbitrarily small.

\begin{theorem}[Proposition 1.1(b) from \cite{PW}]
There is a constant $B_0 >0$ such that any $jJ$-holomorphic map $f: \Sigma \to Z$ with energy less than $B_0$ is a constant map.
\label{minimum energy requirement}
\end{theorem}

\noindent The constant $B_0$ only depends on the complex structure and metric of $\Sigma$. The expository bubbling result is the following theorem:

\begin{theorem}[Bubbling of Holomorphic Curves]
Given a sequence $\{f_s\}$ of $jJ-$holomorphic maps $\Sigma \to Z$ with uniformly bounded energies:

\[E(f_s) := \int_\Sigma e(f_s) < C,\]

\noindent there is a subsequence  still denoted by $\{f_s\}$, a finite set of points $\{p_1,\ldots,p_N\}\subset \Sigma$, and a $jJ-$ holomorphic map $f_0 : \Sigma \to Z$ such that

\begin{enumerate}[label=(\alph*)]
\item $f_s \to f_0$ in $C^1$ on $\Sigma \backslash \{p_1,\ldots,p_N\}$.
\item The energy densities $e(f_s)$ converge as measures to $e(f_0)$ plus a sum of Dirac-delta measures:

\begin{equation}
e(f_s) \to e(f_0) + \sum_{j=1}^N a_j \delta(p_j)
\label{convergence in measure}
\end{equation}

\noindent where $a_j \geq B_0 \;\forall j$. $B_0$ is the constant in Theorem \ref{minimum energy requirement}.

\end{enumerate}
\label{bubbling of maps}

\end{theorem}

\noindent The proof of this theorem contains critical renormalization techniques which induce holomorphic maps from $\mathbb{S}^2$ to $Z$, or "bubbles". We include a sketch here.

\begin{proof}(Sketch)
The subsequence of $\{f_s\}$ exists due to compactness of $\Sigma$, the uniform boundedness of energies, and the Theorem of Arzela-Ascoli. For convenience, we do not change indices when extracting subsequences.

Take the constant $\epsilon_0 > 0$ as in Proposition \ref{Bochner energy estimate}. For each $m \in \mathbb{N}$, it is possible to cover $\Sigma$ with a finite number of discs with radius $r_m=2^{-m}\epsilon_0$. Denote these discs by $\{B(y_\alpha,r_m)\}$. The compactness of $\Sigma$ allows us to further assume that $\{B(y_\alpha,\frac{r_m}{2})\}$ continues to cover $\Sigma$ and that $r_m < inj _{g}$, the {\it injectivity radius} of the chosen background metric. Furthermore, we may assume that each point of $\Sigma$ is covered by these discs at most $M$ times, where $M$ is uniform for all $s$ and $m$. With these choices, and the energy bound, we may conclude that for each $s$,

\[\int_{B(y_\alpha,r_m)} e(f_s) < \epsilon_0\]

\noindent for all but finitely many $m$. By passing to a subsequence we may fix these discs for each $m$ and conclude that

\[\int_{B(y_\alpha,r_m)} e(f_s)  \geq \epsilon_0\]

\noindent only at certain "bad discs" $\{B(p_{1,m},r_m),\ldots,B(p_{N,m},r_m)\}$, where $l$ is uniform over $s$. For each $m$, the estimate \eqref{main energy estimate} holds on all "good discs", namely, those $B(y_\alpha,\frac{r_m}{2})$'s disjoint from bad discs. Since $e(f_s)$ dominates the first derivatives of $f_s$, we see that $f_s$ and their first derivatives are uniformly bounded. By the Theorem of Arzela-Ascoli, we have a subsequence of $\{f_s\}$ that converge in $C^1$ on each good discs. On the other hand, as $m \to \infty$, the centers of bad discs converge (by a subsequence) to $\{p_1,\ldots,p_N\}$. Picking the diagonal subsequence from the double sequence in $s$ and $m$, we conclude the existence of $f_0: \Sigma \backslash \{p_1,\ldots,p_N\} \to Z$ in part (a). The domain of $f_0$ can be extended to $\Sigma$ by Theorem \ref{removable singularity for maps}.

Part (b) concerns energy densities on the bad discs. For each $\epsilon < \epsilon_0$, we may assume that the numbers

\[b_j(s) := \sup_{ B(p_j,\epsilon)} \{|e(f_s)|\}\]

\noindent are unbounded for all $j$. Shrinking $\epsilon$ if necessary, we may assume that $B(p_j,2\epsilon)$ are all disjoint and set

\begin{equation}
a_j := \lim_{\epsilon \to 0}\limsup_{s\to\infty} \int_{B(p_j,\epsilon)} \left||e(f_s)|-|e(f_0)|\right| \omega,
\label{energy loss at x j}
\end{equation}

\noindent the energy loss at each $p_j$ at $s=\infty$. $f_0$ and these $a_j$ then satisfy \eqref{convergence in measure}, and it remains to show that $a_j \geq B_0$.
The inequality is verified by standard renormalization of $f_s$ near each "bubble point" $p_j$. Let $\bar{p}_j^s \in B(p_j,\epsilon)$ be the point at which $|e(f_s)(\bar{p}_j^s)|=b_j(s)$. Then $\bar{p}_j^s \to p_j$ as $s\to\infty$ (up to subsequences, by compactness of $\Sigma$). Fixing an appropriate holomorphic coordinate, we rescale the geodesic ball $B(p_j,\epsilon)$ to define the renormalization

\begin{equation}
\tilde{f}_s(y)=f_s\left(\bar{p}_j^s+\frac{y}{b_j(s)}\right);\;\;\;y\in B(0,\epsilon b_j(s)).
\label{renormalization}
\end{equation}

\noindent It is then clear that $|e(\tilde{f}_s)|\leq 1$ and $|e(\tilde{f}_s)(0)|=1$. Via stereographic projection, we may regard each $\tilde{f}_s$ as a map from a domain of $\mathbb{S}^2 \backslash \{p^+\}$ to $N$, where $p^+$ is the north pole. The energy bound continues to hold as the energy is conformally invariant. We may therefore apply part (a) to extract a subsequence converging to a $jJ$-holomorphic map $\tilde{f}_{p_j}: \mathbb{S}^2 \backslash \{p^+\} \to N$. Removing the singularity we obtain the "bubble map" $\tilde{f}_{p_j}: \mathbb{S}^2 \to N$. Since $e(\tilde{f}_s)(0) \neq 0$, it is not a constant map and therefore $E(\tilde{f}_{p_j}) \geq B_0$ by Theorem \ref{minimum energy requirement}. Part (b) then follows from \eqref{energy loss at x j}.

\end{proof}

The inequality in part (b) of the theorem may be strict. That is, the bubble might not capture all the energy loss $a_j$ at infinity. The main reason for this is that the resacling factor $b_j(s)$ in \eqref{renormalization} may be too large so that too much energy is pushed toward the north pole $p^+$, forming a {\em connecting tube} with positive energy. The tube is removed when extending $\tilde{f}_{p_j}$, resulting in energy loss. To avoid such setback, we adjust the rescaling factor so that the energy on the tubes are controlled by $\epsilon$. Since energies of nonconstant holomorphic maps may not be arbitrarily small, the energies on the connecting tubes must be $0$ as $\epsilon \to 0$(cf. section 5 in \cite{PW}). The modified process introduces the "bubble tree" description of moduli space at infinity, which associates a sequence of $\mathbb{S}^2$ and corresponding holomorphic maps on $\mathbb{S}^2$ at a bubble point. These bubbles will in fact preserve all the original energy. For simplicity, we omit the subscript $j$ of the bubble point $p_j$, as well as all the variables associated to it.

\begin{theorem}[Bubble Tree]
For each bubble point $p$, there exists a finite sequence of holomorphic maps $\{\tilde{f}_p^l\}$ from $\mathbb{S}^2$ to $Z$ so that

\begin{equation}
\sum_l E(\tilde{f}_x^l) = a
\label{bubble tree energy}
\end{equation}

\label{bubble tree map}
\end{theorem}

\noindent The main techniques are outlined below. For each $\epsilon$ and neighborhood $D(p,\epsilon)$ around $p$, we pull back the maps in some carefully designed ways so that they are all defined on domains of two sphere $\mathbb{S}^2_{p} \subset \mathbb{R}^3$. The renormalization takes place there. Each process gives rise to a scenario of Theorem \ref{bubbling of maps}, which results in a collection of bubble points on $\mathbb{S}^2_{p}$ with a possible energy loss concentrated at $p^+$. One repeats the process on each bubble point on $\mathbb{S}^2_p$, producing bubbles on bubbles conserving all the energy.

\begin{proof}(Outline)

Fix $\epsilon > 0$.  Each disc $B(p,\epsilon)$ in the proof of the previous theorem corresponds to a domain $B_\epsilon \subset \mathbb{S}^2_p$ under stereographic projection. Denote the north pole and south pole of $\mathbb{S}^2_p$ by $p^+$ and $p^-$, respectively. The bubble point $p$ then corresponds to the south pole $p^- \in B_\epsilon$. Let also

\begin{equation}
a(\epsilon,s)=\int_{B_\epsilon} \left||e(f_{s})|-|e(f_{0})|\right| \omega,
\label{energy concentation}
\end{equation}

\noindent be numbers so that

\[\lim_{\epsilon \to 0} \limsup_{s \to \infty} a(\epsilon,s) = a \geq B_0\]

\noindent as in Theorem \ref{bubbling of maps}.

We consider the following composition of conformal maps:

\begin{equation}
R_{\epsilon,s}: \mathbb{S}^2_p \xrightarrow{\rho_{t_{\epsilon,s}}} \mathbb{S}^2_p \xrightarrow{T_{\epsilon,s}} \mathbb{S}^2_p \xrightarrow{\sigma} T_p\Sigma \xrightarrow{exp} \Sigma.
\label{composition for renormalization}
\end{equation}

\noindent  Here, $exp$ and $\sigma$ are the ordinary exponential map and  stereographic projection, respectively. $T_{\epsilon,s}$ is a conformal transformation on $\mathbb{S}^2_p$ corresponding to the translation of $T_p\Sigma$ that translates the center of mass of the measure $\left||e(f_s)|-|e(f_0)|\right|$ to $z-$axis. Finally, $\rho_{t_{\epsilon,s}}$ is the conformal transformation corresponding to radial dilation of $T_x\Sigma$ by $t_{\epsilon,s} >0$.  Let $C_0>0$ be a constant less than $\frac{B_0}{2}$, where $B_0$ is  the lower bound of energies of nonconstant holomorphic maps described in Theorem \ref{minimum energy requirement}. The scale $t_{\epsilon,s}$ is chosen fso that

\begin{equation}
\int_{B_\epsilon^s\backslash H^-} \left||e(f_s)|-|e(f_0)|\right| = C_0,
\label{energy on northern hemisphere}
\end{equation}

\noindent where $H^-$ is the southern hemisphere and $B_\epsilon^s := R_{\epsilon,s}^* B_(p,\epsilon)$. In another words, a constant amount of energy is retained on the northern hemisphere throughout the process. Such a scale is possible by continuous dilation to continuously spread out the energy concentration at $p$ ( or $p^-$). It is also necessary that $t_{\epsilon,s} \to \infty$ as $s \to \infty$. Indeed, for each $s$, let $\bar{p}^s \in B(p,\epsilon)$ be the point where $|e(f_s)|$ achieve its supremum. These points converge to the bubble point $p$ as $s \to \infty$ and energies of $|e(f_s)|$ are arbitrarily concentrated near $p^-$. Therefore, for \eqref{energy on northern hemisphere} to be true, the scaling factors $t_{\epsilon,s}$ must be large enough so that $t_{\epsilon,s} \bar{p}^s \nrightarrow p$, which requires $t_{\epsilon,s} \to \infty$ so that a constant amount of energy to be kept away from $p^-$. The renormalized map is then

\begin{equation}
\tilde{f}_{\epsilon,s} := R_{\epsilon,s}^* (f_s|_{B(p,\epsilon)}),
\label{renormalized map bubble tree}
\end{equation}

\noindent which are holomorphic with respect to the pull back complex structures $j_s=R_{\epsilon,s}^*J$. The structures approach the standard complex structure on $\mathbb{S}_p^2$ as $s\to\infty$. One may readily verify that the domains of $\tilde{f}_{\epsilon,s}$ approach $\mathbb{S}^2_p \backslash \{p^+\}$ as $s \to \infty$. Moreover, with conformal invariance, we have $E(\tilde{f}_{\epsilon,s}) \leq C$,

\begin{equation}
|E(\tilde{f}_{\epsilon,s})-E(\tilde{f}_{\epsilon,0})|> \frac{B_0}{2},
\label{conformal invariance 1}
\end{equation}

\noindent and

\begin{equation}
\int_{H^+}\left||e(\tilde{f}_{\epsilon,s})|-|e(\tilde{f}_{\epsilon,0})|\right| = C_0,
\end{equation}

\noindent where $H^+$ is the northern hemisphere amd $\tilde{f}_{\epsilon,0}=\lim_{s\to\infty}(R_{\epsilon,s}^* f_0)$.  We may repeat the arguments in the proof of Theorem \ref{bubbling of maps} and obtain the list of bubbling points

\begin{equation}
\mathbb{B}_\epsilon:=\{y_{1,\epsilon},\ldots,y_{l,\epsilon},p^+\} \subset \mathbb{S}^2_x
\label{list of bubble point epsilon}
\end{equation}

\noindent and a $jJ-$holomorphic map $\tilde{f}_{\epsilon,p}:\mathbb{S}^2_p \to Z$ so that $\tilde{f}_{\epsilon,s} \to \tilde{f}_{\epsilon,x}$ in $C^1$ on $\mathbb{S}^2_p \backslash \mathbb{B}_{\epsilon}$. The adjusted renormalization satisfies part (b) of Theorem \ref{bubbling of maps}: For all $\epsilon$,

\begin{equation}
e(\tilde{f}_{\epsilon,s}) \to e(\tilde{f}_{\epsilon,x}) + \sum_{j=1}^l a_{j,\epsilon}  \delta(y_{j,\epsilon}) + \tau_{\epsilon,p} \delta(p^+)
\label{enregy loss renormalized}
\end{equation}

\noindent as $s\to\infty$, where $\tau_{\epsilon,p}$ is the energy loss at infinity.

Finally, we shrink the radius $\epsilon$ of the initial disc around $p$. Pick a sequence $\epsilon_s \to 0$ as $s \to \infty$. The bubble points and associated energy losses

\begin{equation}
\{y_{1,\epsilon_s},\ldots,y_{l,\epsilon_s},a_{1,\epsilon_s},\ldots,a_{l,\epsilon_s},\tau_{\epsilon_s,p}\}_s.
\label{list of bubble point}
\end{equation}

\noindent range in compact sets and possess subsequences converging to

\begin{equation}
\mathbb{B}:=\{y_1,\ldots,y_l,a_1,\ldots,a_l,\tau_p\}
\label{list of bubble point limit}
\end{equation}

\noindent as $s\to \infty$. The holomorphic maps $\tilde{f}_{\epsilon_s,p}$ correspondingly converge  to a map $\tilde{f}_p : $ in $C^1$ on $\mathbb{S}^2_p  \backslash \mathbb{B}$ which extends to the entire $\mathbb{S}^2_p$ holomorphically. Since $f_0$ is smooth on $B(p,\epsilon)$, by conformal invariance we have $e(\tilde{f}_{\epsilon,0}) \to 0$ as $\epsilon \to 0$ in measure. Consequentially, we have

\begin{equation}
e(\tilde{f}_{\epsilon_s,s}) \to e(\tilde{f}_p) + \sum_{n=1}^l a_n \delta(y_n)+\tau_p \delta (p^+)
\label{convergence in shrinking disc 1}
\end{equation}

\noindent and

\begin{equation}
\int_{H^+} |e(\tilde{f}_{\epsilon_s,s})| \to C_0
\label{convergence in shrinking disc 2}
\end{equation}

\noindent as $s \to \infty$. Shrinking the constant $C_0$ if necessary, Lemma 5.3 of \cite{PW} then implies that $\tau_p =0$. This is a property special to holomorphic maps for which the energy loss is the area of the image of certain domain $\Omega_{\epsilon_s}$ containing $p^+$ under $\tilde{f}_{\epsilon,s}$ known as the {\em connecting tube}. The holomorphicity further ensures that the area of $\tilde{f}_p(\Omega_{\epsilon_s})$ is bounded by its perimeter which is controlled by $\epsilon_s$, with an isoperimetric inequality. We therefore have the conservation of energy

\begin{equation}
a=\sum_{n=1}^l a_n.
\label{energy conservation}
\end{equation}

The enregy-preserving renormalization may be iterated. Near each $y_n \in \mathbb{B}$, we may renormalize $\tilde{f}_p$ to obtain a collection of bubble points and bubble energies

\[\mathbb{B}_n^1:=\{y_{n,1},\ldots,y_{n,l_n},a_{n,1},\ldots,a_{n,l_n}\}\]

\noindent on another sphere $\mathbb{S}^2_{y_n}$ equipped with a "bubble" $\tilde{f}_{y_n}: \mathbb{S}^2 \to N$ constructed identically as above. Of course, these new energies satisfy

\[a_n = \sum_{t=1}^{l_n} a_{n,t}.\]

\noindent The process continues and we end up with $\Sigma$ and a collection of bubbles wedging at various bubble points, forming a bubble tree.

\end{proof}

The origins of new bubbles $y_{1,\epsilon},\ldots,y_{l,\epsilon}$ in \eqref{list of bubble point epsilon}, as well as their $\epsilon$ limit \eqref{list of bubble point limit}, can be explicitly explained. The only bubble point on $B(p,\epsilon)$, namely $p$, forms due to the accumulation of $\bar{p}^s$, points where $|e(f_s)|$ achieves supremum. The renormalized map $\tilde{f}_{\epsilon,s}$ then has maximum energy density at $R_{\epsilon,s}^{-1}(\bar{p}^s)$. Therefore, the only possibilities for new bubbles to form upon renormalization is the presence of new limit points from the sequence $\{R_{\epsilon,s}^{-1}(\bar{p}^s)\}$ in $\mathbb{S}_p^2$. Since the finite energy condition is invariant under $R_{\epsilon,s}$, there are only finitely many such points $y_{1,\epsilon},\ldots,y_{l,\epsilon}$ and so are their $\epsilon$-limit points and corresponding bubbles.

It is important to point out that the limiting map $\tilde{f}_p :\mathbb{S}^2_x \to Z$ may itself be a constant map and $E(\tilde{f}_p)=0$.  That is, the energy of the bubble is entirely concentrated at the new bubble point(s) $\mathbb{B}$ on it. We refer to such sphere as a {\em ghost bubble}. A ghost bubble for the non-constant holomorphic map $\tilde{f}_p$, however, must contain at least 2 new bubble points.

\begin{lemma}[cf. Lemma 4.2 in \cite{PW}]

  For $\tilde{f}_p$ in Theorem \ref{bubble tree map}, if $E(\tilde{f}_p)=0$, the integer $l$ in $\mathbb{B}$ is at least 2.
\label{ghost bubble}
\end{lemma}

All these newly induced objects extend both the domain of original maps and give rise to a new map $\tilde{f}$ on the extended domain consisting of a finite number of $\mathbb{S}^2$'s wedging at bubble points. We denote the union of these spheres by $\bar{\mathcal{T}}$. A more geometrical description of $\bar{\mathcal{T}}$ is through a tower of $\mathbb{S}^2$-fibrations. A complex structure on $\Sigma$ determines the complex tangent bundle $T\Sigma := F\Sigma \times_{\mathbb{C}^*} \mathbb{C}$, where $F\Sigma$ is the complex frame bundle. Stereographic projection $\sigma$, an orientation preserving conformal map $\mathbb{S}^2 \to \mathbb{R}^2$, compactifies $T\Sigma$ into $S\Sigma := F\Sigma \times_{\mathbb{C}^*} \mathbb{S}^2$. The newly formed bubble points $\{y_1,\ldots,y_l\}$ are therefore elements in the fiber of $S\Sigma$ over bubble point $p$. Renormalization on $y_n$'s therefore give rise to more bubble points on $SS\Sigma$, and so forth. We therefore have a tower of sphere fibrations

\begin{equation}
\cdots \to S^m \Sigma \to S^{m-1}\Sigma \to \cdots \to S\Sigma \to \Sigma ,
\label{tower of fibrations}
\end{equation}

\begin{definition} [\cite{PW}]
A bubble domain $B$ at level $m$ is a fiber ($\simeq\mathbb{S}^2$) at the level $m$ in \eqref{tower of fibrations}. A bubble domain tower is an extension of $\Sigma$ by a finite union of bubble domains:

\begin{equation}
\bar{\mathcal{T}}:=\Sigma \bigcup_{m=1}^N T_m,
\label{bubble domain tower}
\end{equation}

\noindent where each $T_m$ is a finite collection of bubble domains at level $m$. A {\em bubble tree} with bubble domain $\mathcal{T}$ is $\bar{\mathcal{T}} / \sim$, identifying the north pole of each bubble domain (fiber) with its base point.
\label{bubble domain and tree}
\end{definition}

\noindent $\mathcal{T}$ has a natural tree structure, where vertices are maps and bubble points form the incident edges.  The discussions above associate to every sequence of holomorphic maps a unique bubble tree $\mathcal{T}$ where $\Sigma$ and every sphere is mapped holomorphically into $Z$. The image of $\mathcal{T}$ under $\tilde{f}$ is known as a cusp curve (\cite{W}).

Gromov compactness is constructed in this extended scope. For each $s$, we may extend the domain of $f_s$ from $\Sigma$ to $\mathcal{T}$ by appropriate surgeries. The precise definition requires the following lemma.

\begin{lemma}[Extension Lemma 6.1 of \cite{PW}]
For each $A>0$, there is $\epsilon_A >0$ such that for all $\epsilon < \epsilon_A$ and continuous $L^{1,2}$ map

\[f: \Sigma \backslash B(x,\epsilon) \to N\]

\noindent with $E(f)<A$ extends to a continuous $L^{1,2}$ map $\bar{f}: \Sigma \to N$.

Moreover, the following estimate hold on $B_f=B(x,r_f)$ with $\epsilon < r_f < \sqrt{\epsilon}$:

\[\int_{B_f} |d \bar{f}|^2 \leq C |\log \epsilon|^{-1}\]

\noindent and

\[dist(\bar{f}(y),f(z)) \leq C |\log \epsilon^{-\frac{1}{2}}|\;\;\forall y \in B_f, \; z\in\p B_f.\]

\label{extension lemma}
\end{lemma}

\noindent With the lemma,  we extend the domain of $f_s$, or construct the \emph{prolongation} of $f_s$ on $\mathcal{T}$. For each holomorphic $f_s:\Sigma \to Z$ considered above, we define

\[\mathcal{P}_\epsilon (f_s): \mathcal{T} \to Z\]

\noindent separately on $\Sigma$ and the spheres as follows.

First restrict $f_s$ on $\Sigma \backslash \cup B(p,\epsilon)$ for $\epsilon$ small enough, and extend $f_s$ across the small discs around bubble points to a map $\bar{f}_s$ according to Lemma \ref{extension lemma}. We define

\[\mathcal{P}_\epsilon (f_s) (z) = \bar{f}_s (z)\;\; ; \forall z\in\Sigma.\]

\noindent Let now $z \in \mathbb{S}^2_p$, a bubble attached to a bubble point $p \in \mathcal{T}$. Each $\epsilon >0$ is associated to a disc $B_\epsilon= B(p^+,\epsilon)$ around the north pole so that for large enough $s$, the renormalized $\tilde{f}_{\epsilon,s}$ is defined outside $D_\epsilon$. The map $\tilde{f}_{\epsilon,s}$ has its bubble points $\{y_j\}$, and let $B_j := B(y_j,\epsilon)$. We then restrict $\tilde{f}_{\epsilon,s}$ on $\mathbb{S}^2 \backslash \cup_j B_j$ and extends again by Lemma \ref{extension lemma} to $\hat{f}_p : \mathbb{S}^2_p \to Z$. We define

\[\mathcal{P}_\epsilon (f_s) (z) = \hat{f}_p (z) \;\; ; \forall z\in\mathbb{S}^2_p.\]

Identical definitions apply on neighborhoods around other bubble points. We then have the rigorous sense of Gromov compactness of sequence of holomorphic maps with bounded energies.

\begin{theorem}[Gromov Compactness, Theorem 6.2 of \cite{PW}]
Let $\{f_s\}$ be a sequence of $jJ-$ holomorphic maps $\Sigma \to Z$. Then there is a bubble tree $\mathcal{T}$, and a sequence $\epsilon_s \searrow 0$ as $s\to \infty$ such that a subsequence of

\[\mathcal{P}_{\epsilon_s}(f_s) : \mathcal{T} \to Z\]

\noindent converges in $C^0 \cap  L^{1,2}$ to a $jJ-$holomorphic map $f: \mathcal{T} \to Z$. The convergence is in $C^r(K)$ for all compact set $K$ away from bubble points.

\label{Gromov compactness for maps}
\end{theorem}

\noindent We say that such $f_s$ \emph{Gromov converge} to $f$ on the bubble tree $\mathcal{T}$.

\section{Bubbling of Vortex Moduli Spaces}

We now return to vortex moduli spaces. The bubbling phenomenon we study takes place on a family of vortices $\{[D_s,\phi_s]\} \subset \nukzero$ approaching the boundary of $\nuk$ as $s \to \infty$. By Theorem \ref{generalized moduli space}, each $[D_s,\phi_s]$ projects down to a tuple of linearly equivalent divisors $(E_{s,i})_{i=0}^k \in \mathbb{E}_{k+1,0}$ that approaches the boundary of $\mathbb{E}_{k+1}$. Convergence to the boundary indicates the coalescence of these divisors at infinity, or common zeros of these tuples of sections. Precisely, we write each divisor $E_{s,i}$ into a sum of two divisors

\begin{equation}
E_{s,i} = \sum_{j=1}^N a_{ij}(s) p_{ij}(s) + \sum_{l=1}^{N_i(s)} b_{il}(s) q_{il}(s),
\label{divisor sum}
\end{equation}

\noindent where the integral coefficients above sum up to $r$. The first sum corresponds to coalescence, that is, $a_{ij}(s) \to a_j$ and $p_{ij}(s) \to p_j$ as $s\to \infty$ for all $i$. We let

 \[E := \sum_{j=1}^N a_j p_j.\]

\noindent The second sum then consists of points that remain separated as $s \to \infty$. That is, $b_{il}(s) \to b_{il}$, $N_{si}\to N_i$, $q_{il}(s) \to q_{il}$ as $s\to \infty$ and

\[\bigcap_i \{q_{il}\} = \emptyset.\]

\noindent The divisors

\[E_{0 i} := \sum_{l=1}^{N_i} b_{il} q_{il} \]

\noindent give rise to new holomorphic maps with degree $r-\sum_j a_j$. Note that $a_j(s)$, $N_{i}(s)$, and $b_{il}(s)$ in \eqref{divisor sum} are all integers and may be assumed constants in $s$.

On the other hand, the vortices $\{[D_s,\phi_s]\} \subset \nukzero$ correspond to degree $r$ maps $\{f_s\}\subset \Hol$ via the diffeomorphism $\Phi_s$ described in section 3:

\[f_s(p) := [\phi_{s,0}(p) : \cdots : \phi_{s,k}(p)].\]

\noindent Equip $\pk$ with the Fubini-Study metric, we attempt to explicitly express the energy density of $f_s$. For each $s$, we observe the energy density $e(f_s)$ defined by

\begin{equation}
e(f_s) \omega = f_s^* \omega_{FS} = \p\bar{\p} \log \left(\sum_{i=0}^k |\phi_{s,i}|^2\right).
\label{energy density equation}
\end{equation}

\noindent $f_s$'s are of uniformly bounded, in fact constant, energies:

\[E(f_s)=\int_\Sigma e(f_s) \omega = \frac{1}{2\pi} r \;\;\forall s\]

\noindent and they fit into the discussions of section 4. In particular, $e(f_s)$'s blow up at finitely many points. Straightforward computations show that the only possible blow up points are $p_1,\ldots,p_N \in supp(E)$. For each $j$, fix a normal neighborhood $B_j:=B(p_j,\epsilon)$ small enough so that $B_j \cap_{s,i} E_{s,i} = \{p_{ij}(s)\}$. Given a local trivialization, each section $\phi_{s,i}$ is locally given by

\begin{equation}
\phi_{s,i} = (z-p_{ij}(s))^{a_j}f_{s,i},
\label{local expression of section}
\end{equation}

\noindent where $p_{ij}(s) \to 0$ as $s\to \infty$ and $f_{s,i}$ are non-vanishing holomorphic functions on $B_j$. Moreover, $f_{s,j}$ converge smoothly to a non-vanishing holomorphic function on $B_j$ by Theorem \ref{Main Theorem L}. With respect to this trivilization, the globally defined $(1,1)$ form \eqref{energy density equation} is then locally given by

\begin{eqnarray}
& &f_s^* \omega_{FS} \nonumber \\
&=& e(f_s) dz\wedge d\bar{z} \nonumber \\
&=&\frac{\sqrt{-1}}{2\pi} \left(\frac{\sum_{i=0}^k |z-p_{ij}(s)|^{2a_j-2}G_{s,i}\sum_{i=0}^k |z-p_{ij}(s)|^{2a_j}|f_{s,i}|^2}{\left[\sum_{i=0}^k |z-p_{ij}(s)|^{2a_j}|f_{s,i}|^2\right]^2}\right)  \; dz\wedge d\bar{z}\nonumber \\
& &- \frac{\sqrt{-1}}{2\pi}\left(\frac{\sum_{i=0}^k |z-p_{ij}(s)|^{2a_j-2}(z-p_{ij}(s))F_{s,i}\sum_{i=0}^k |z-p_{ij}(s)|^{2a_j-2}\overline{(z-p_{ij}(s))}H_{s,i}}{\left[\sum_{i=0}^k |z-p_{ij}(s)|^{2a_j}|f_{s,i}|^2\right]^2}\right)\; dz\wedge d\bar{z}, \nonumber \\
\label{energy density in coordinate}
\end{eqnarray}

\noindent where $F_{s,i},G_{s,i},$ and $H_{s,i}$ are smooth and nonvanishing functions on $B_j$ consisting of $f_{s,i}$ and its derivatives. They converge in $C^\infty$ and the only sources of singularities are $|z-p_{ij}(s)|$'s. It is then clear that the bubbling behaviours depend crucially on the convergence behaviours of $p_{ij}(s)$ to $0$ as $s \to \infty$. We first observe the outcome when singularities are formally ignored:

\begin{theorem}[Formal Removal of Singularities]
Fix a Hermitian line bundle $(L,H)$ over $\Sigma$. Given a sequence of vortices $\{[D_s,\phi_s]\}\subset \nukzero$ approaching the boundary of $\nuk$, there exists a finite set of points $\{p_1,\ldots,p_N\}\subset \Sigma$, integers $\{a_1,\ldots,a_N\}\subset \mathbb{N}$ such that $\sum_ja_j \leq r$, and vortices $[D_s^\prime,\phi_s^\prime]$ with smooth (subsequential) limit $[D_0,\phi_0]$ on line bundle

\[L_0:=L \otimes_j \mathcal{O}(-a_j p_j),\]

\noindent such that

\begin{itemize}
\item $[D_s^\prime,\phi_s^\prime] = [D_s,\phi_s]$ on $\Sigma \backslash \{p_1,\ldots,p_N\}$ (via the isomorphism $L_0 \simeq L$ on $\Sigma \backslash \{p_1,\ldots,p_N\}$).
\item $D_s^\prime$ and $\phi_s^\prime$ satisfy the vortex equation
\begin{equation}
\begin{cases}
  D_s^{\prime (0,1)} \phi_{s,i}^\prime = 0 \;\;\forall i \\
  \sqrt{-1}\Lambda F_{D_s^\prime} + \frac{s^2}{2}\left(\sum_{i=0}^k|\phi_{s,i}^\prime|^2_H-1\right)=0.   \\
  \end{cases}
  \label{s vortex prime}
\end{equation}
on $L_0 \to \Sigma$.
\item $[D_0,\phi_0]$ satisfies
\begin{equation}
\begin{cases}
  D_0^{(0,1)} \phi_{0,i} = 0 \;\;\forall i \\
  \sum_{i=0}^k|\phi_{0,i}|^2_H-1=0.   \\
  \end{cases}
\label{infity vortex after removal of singularity}
\end{equation}
on $L_0 \to \Sigma$.
\end{itemize}
\label{Removal of Singularities}
\end{theorem}

\begin{proof}
We continue the usage of notations introduced in \eqref{divisor sum} $\sim$ \eqref{energy density in coordinate}.

For large enough $s$, a family $\{[D_s,\phi_s]\}$ is uniquely associated with a family of tuples $\{(E_{s,0},\ldots,E_{s,k},\tau_s)\}$, where $\tau_s \in (\mathbb{C}^*)^{k+1} / \mathbb{C}^*$ as in Theorem \ref{generalized moduli space} and $E_{s,i}$ as in \eqref{divisor sum}:

\[E_{s,i} = \sum_{j=1}^N a_j p_{j}^s + \sum_{l=1}^{N_i} b_{il} q_{il}^s.\]

Let $\psi_j$ be the defining meromorphic section of $\mathcal{O}(-a_j p_j)$ and consider the holomorphic sections $\phi_s^{'} = (\phi_{s,i}^\prime)_i$, defined by

\begin{equation}
\phi_{s,i}^\prime := \phi_{s,i} \otimes_j \psi_j \in H^0(\Sigma,L_0).
\label{bubbled off section}
\end{equation}

\noindent Choose $\psi = 1$ away from $p_j$'s so that on every compact subset $K \subset \Sigma \backslash \{p_1,\ldots,p_N\}$, $L|K \simeq L_0|K$ and $\phi_{s,i}^\prime=\phi_{s,i}$. The sections $\phi_{s,i}^\prime$ do not have common zero and define a degree $r-l$ holomorphic map $f_s^\prime: \Sigma \to \pk$, where $l=\sum_j a_j$ and $(f_s^\prime)^* \mathcal{O}(1) \simeq L_0$ smoothly. $[D_s^\prime,\phi_s^\prime]$ are then vortices defined by $f_s^\prime$ that satisfy \eqref{s vortex prime} via  identical construction of $\Phi_s$ in \eqref{correspondence maps vortices}. Since $f_s^\prime = f_s$ on $\Sigma \backslash \{p_1,\ldots,p_N\}$, the first and second statements of the theorem are clear from our constructions.

Staying in the compact region of the moduli space, the limit of $[D_s^\prime,\phi_s^\prime]$ is then naturally constructed from the limit of $f_s^\prime$.  The energy densities $e(f_s^\prime)$ are uniformly bounded on the entire $\Sigma$ since they are defined by sections whose zeros $\sum_{l=1}^{N_i} b_{il} q_{il}^s$ do not coalesce as $s \to \infty$. Let $f_0 \in \mathcal{H}_{r-l,k}$ be their limit, which defines a vortex $[D_0,\phi_0]$ on $L_0$. It is clear that $[D_s^\prime,\phi_s^\prime] \to [D_0,\phi_0]$ in $C^1$.

To verify that $[D_0,\phi_0]$ satisfies \eqref{infity vortex after removal of singularity}, we note that holomorphic condition from first equation of \eqref{s vortex prime} continues to hold as $s \to \infty$ via standard elliptic regularity arguments. As for the second equation, we note that $\sqrt{-1}\Lambda F_{D_s^\prime}$ are uniformly bounded in $s$. Indeed, $F_{D_s^\prime}$ is the pullback of curvature form $F_{FS}$ on $\mathcal{O}(1)$, which is proportional to $\omega_{FS}$, via $f_s^\prime$. Therefore, $\sqrt{-1}\Lambda F_{D_s^\prime} = C e(f_s^\prime)$ for all $s$, which are uniformly bounded by the discussions above. Dividing the second equation of \eqref{s vortex prime} by $\frac{s^2}{2}$ and let $s \to \infty$, the proof is completed.

\end{proof}

Evidently, formal removal of singularities reduces the topological degree of $L$ by $l=\sum_j a_j$. The loss is a clear consequence of concentration of energy densities of $f_s$. In another words, curvature forms corresponding to $f_s$ approach a smooth form plus a Dirac delta current supported on these isolated singularities, and the extension simply ignore the singular part. These vortices $[D_s^\prime,\phi_s^\prime]$ corresponds to the map defined by

\[f_s^\prime(z) = [f_{s,0}(z):\cdots:f_{s,k}(z)]\]

\noindent from \eqref{local expression of section}, $e(f_s^\prime)=W_s$ converges to the smooth function $W_0=e(f_0)$ as $s \to \infty$.

As in Theorem \ref{bubbling of maps} and \ref{bubble tree map}, we apply the carefully designed renormalization process to smooth out the singularities and use the limiting map to obtain the limiting vortex preserving the energy ignored in Theorem \ref{Removal of Singularities} by bubbles. Part of the following theorem is in fact a specialization of observations from \cite{GS}.

\begin{theorem}[Renormalization]
For each $p_j$ in Theorem \ref{Removal of Singularities}, there exists $\epsilon_s \to 0$ as $s\to\infty$ so that the geodesic disc $B(p_j,\epsilon_s)$ is conformally equivalent to $B_s \subset \mathbb{S}^2$, an increasing family of domains with $\cup_s B_s = \mathbb{S}^2 \backslash \{p^+\}$, and  the followings hold:

\begin{itemize}
\item The pullbacked vortices $[D_s^*,\phi_s^*]$ on $B_s$, satisfies
\begin{equation}
\begin{cases}
  D_s^{*0,1} \phi_{s,i}^* = 0 \\
  \sqrt{-1}\Lambda_s^* F_{D_s^*}+\frac{s^2}{2t_j(s)^2}(\sum_{i=0}^k|\phi_{s,i}^*|^2_H-1)=0\\
  \end{cases}
  \label{pullbacked s vortex}
\end{equation}
\noindent on line bundle $L_s^* := R_j(s)^* L$ over $B_s$. They coincide with the vortices defined by holomorphic functions $\tilde{f}_s := R_j(s) \circ f_s: B_s \to \pk$ in the way of \eqref{correspondence maps vortices}. Conformal maps $R_j(s)$ and parameters $t_j(s) \to \infty$ have been introduced in the proof of Theorem \ref{bubble tree map}. (Here, we denote $R_{\epsilon_s,s}, t_{\epsilon_s,s}$ by $R_j(s)$ and $t_j(s)$.)
\item Exactly one of the followings holds true:

\begin{enumerate}[label=(\alph*)]
\item There exists a $C^1_{loc}$-convergent subsequence of $\{[D_s^*,\phi_s^*]\}$ whose limit $[D_j,\phi_j]$ satisfies

\begin{equation}
\begin{cases}
  D_j \phi_{j,i} = 0 \;\;\forall i\\
 \sum_{i=0}^k|\phi_{j,i}|^2_H-1=0\\
  \end{cases}
  \label{pullbacked infinity vortex holomorphic sphere}
\end{equation}

defined on the entire $\mathbb{S}^2$. That is, a holomorphic sphere in $\pk$ bubbles off.

\item There exsits points $\{p_j^1,\ldots,p_j^{N_j}\}\subset \mathbb{S}^2$, nonnegative integers $a_j^0,a_j^1,\ldots,a_j^{N_j}$, and a $C^1_{loc}$-convergent subsequence of $\{[D_s^*,\phi_s^*]\}$ on $\mathbb{S}^2 \backslash \{p_j^1,\ldots,p_j^{N_j},p^+\}$, whose limit $[D_j,\phi_j]$ satisfies
\begin{equation}
\begin{cases}
  D_j \phi_{j,i} = 0 \;\;\forall i\\
  \sqrt{-1}\Lambda^* F_{D_j}+\frac{1}{2}(\sum_{i=0}^k|\phi_{j,i}|^2_H-1)=0\\
  \end{cases}
  \label{pullbacked infinity vortex}
\end{equation}
\noindent on a degree $a_j^0$ line bundle $L_j$ over $\mathbb{S}^2 \backslash \{p_j^1,\ldots,p_j^{N_j},p^+\}$. Moreover, $(L_j,D_j,\phi_j)$ is the $C^1$ limit of $(L_s^*,D_s^*,\phi_s^*)$.

\end{enumerate}
\item On $\mathbb{S}^2$, $\sqrt{-1}\Lambda^*F_{D_j}$ is a distribution given by a smooth function plus $\sum_{l=1}^{N_j} a_j^l \delta(p_j^l)$.
\end{itemize}
\label{renormalization}
\end{theorem}

We make a brief digression to observe the relationship between the scales $t_j(s)$ and zeros $p_{ij}(s)$ for the simple case when $B(p_j,\epsilon)$ is Euclidean. Moreover, we assume that all $f_{s,i}$'s, the nonvanishing parts of the sections in \eqref{local expression of section}, are $1$:

\begin{equation}
  \phi_{s,i}=(z-p_{ij}(s))^{a_j}
  \label{local expression of sections simple}
\end{equation}

\noindent and $f_s(z) = [(z-p_{oj}(s))^{a_j}:\cdots:(z-p_{kj}(s))^{a_j}]$. The extension map $f_0$ in Theorem \ref{removable singularity for maps} is therefore constant and of zero energy density. The energy density is simplified considerably:

\begin{eqnarray}
  & &e(f_s) \nonumber \\
  &=& e(f_s)-e(f_0) \nonumber \\
  &=& \frac{\sqrt{-1}}{2\pi} a_j \frac{\sum_{i=0}^k |z-p_{ij}(s)|^{2a_j-2}\sum_{i=0}^k |z-p_{ij}(s)|^{2a_j}}{\left(\sum_{i=0}^k |z-p_{ij}(s)|^{2a_j}\right)^2} \nonumber \\
  &-& \frac{\sqrt{-1}}{2\pi} a_j\frac{\sum_{i=0}^k |z-p_{ij}(s)|^{2a_j-2}(z-p_{ij}(s))\sum_{i=0}^k |z-p_{ij}(s)|^{2a_j-2}\overline{(z-p_{ij}(s))}}{\left(\sum_{i=0}^k |z-p_{ij}(s)|^{2a_j}\right)^2} \nonumber \\
  \label{energy density simple case}
\end{eqnarray}

\noindent We denote by $e(f_{s,t})$ the energy density of $f_s$ pullback back by dilation $t$. By conformal invariance, we have

\begin{equation}
  \int_{B(0,t\epsilon)} e(f_{s,t}) dy_t^2 = a_j(\epsilon,s)
\end{equation}

\noindent for all $t$ and $a_j(\epsilon,s) \to a_j$ as $\epsilon \to 0$ and $s \to \infty$. Recall that $t_j(s)$ is the scaling factor so that the masses of $e(f_{s,t})$ are $C_0$ for all $s$ on the annulus $B(0,t_j(s)\epsilon) \backslash B(0,1)$ (which correspond to the $\sigma^{-1}(B(0,t_j(s)\epsilon)) \cap H^+$, where $\sigma$ is the stereographic projection.) Also recall that $C_0$ is strictly less than half of the lower bound of energies of nonconstant holomorphic maps between $\mathbb{S}^2$ and $\pk$. Consider

\begin{eqnarray}
  F_s(t)&:=& \int_{B(0,1)} e(f_{s,t})  dy_t^2  \nonumber \\
&=&\int_0^{2\pi} \int_0^{\frac{1}{t}}  r e(f_s)(r,\theta) \; drd\theta. \nonumber \\
\end{eqnarray}

\noindent The above conditions for $t_s$ then require that

\begin{equation}
  F_s(t_s) = a_j(\epsilon,s)-C_0.
  \label{constant energy requirement}
\end{equation}

\noindent Differentiating $F_s(t)$ with respect to $t$ and applying fundamental theorem of calculus, we have

\begin{equation}
  F_s^\prime(t) = \frac{K_s}{t},
  \label{differentiation of F s t}
\end{equation}

\noindent where

\begin{equation}
  K_s = \int_0^{2\pi} -e(f_s)|_{\mathbb{S}^1} d\theta.
  \label{definition of K s}
\end{equation}

\noindent $t$ is absent in $K_s$ since $e(f_{s,t})$ = $t^2 k_s(z)$, and $t^2$ cancels with the derivative of $\frac{1}{t}$.  $k_s$ is independent of $t$ and it is straightforward to check that $K_s \to 0$ as $s \to \infty$ since $k_s$ does. We have

\begin{equation}
  t_s= e^{\frac{a_j(\epsilon,s)-C_0-C}{K_s}}.
  \label{formula for t s}
\end{equation}

\noindent for some $C>0$. Since $e(f_s)$ concentrates near $0$, we see that for $s$ enough, $F_s(1)=C=a_j$. Also note that $t_s \to \infty$ as $s \to \infty$ as expected.

\noindent For the general case of metric $\omega$ and non-vanishing holomorphic functions $f_{s,i}$, we note that $\omega_t$ is asymptotically flat and the functions approach non-vanishing holomorphic functions. \eqref{formula for t s} is then asymptotically satisfied.

We begin the proof of Theorem \ref{renormalization}.

\begin{proof}

With the preparations above, we repeat the proof of Theorem \ref{bubble tree map} except removal of singularity. Recall the composition of conformal transformations in \eqref{composition for renormalization}:

\begin{equation}
R_j(s): \mathbb{S}^2 \xrightarrow{\rho_{t_j(s)}} \mathbb{S}^2 \xrightarrow{T_s} \mathbb{S}^2 \xrightarrow{\sigma} T_{p_j}\Sigma(\simeq \mathbb{C}) \xrightarrow{exp} B(p_j,\epsilon),
\label{composition for renormalization vortices}
\end{equation}

\noindent where $T_s$ and $\rho_{t_j(s)}$ are the conformal transformations on $\mathbb{S}^2$ corresponding to appropriate translations and dilations $y \to \frac{y}{t_j(s)}$ on $\mathbb{C}$, respectively. $\sigma$ and $exp$ are the usual stereographic projection and exponential map, respectively. Let $\tilde{f}_s := f_s \circ R_j(s)$ and $B_s := R_j(s)^{-1}(B(p_j,\epsilon))$, the following diagram is considered:

\begin{equation}
\begin{tikzpicture}
  \matrix (m) [matrix of math nodes,row sep=3em,column sep=4em,minimum width=2em]
  {
     L_s^* := \tilde{f}_s^*\mathcal{O}(1) & L & (\mathcal{O}(1),H_{FS}) \\
     B_s & B(p_j,\epsilon) & \pk \\};
  \path[-stealth]
    (m-1-1) edge  (m-2-1)
    (m-2-1.east|-m-2-2) edge node [above] {{\tiny $R_j(s)$}} (m-2-2)
            edge [bend right=20] node [below] {{\tiny $\tilde{f}_s$}} (m-2-3)
    (m-1-2) edge  (m-2-2)
    (m-2-2) edge node [above] {{\tiny $f_s$}} (m-2-3)
    (m-2-3) edge node [right] {{\tiny $z_0,\ldots,z_k$}}(m-1-3);
\end{tikzpicture}
\label{pull back diagram}
\end{equation}

\noindent The complex structure on $B_s$ is given by

\begin{equation}
  J_s(\xi)=j_s(\xi) J_{\mathbb{S}^2},
  \label{pullback complex structure}
\end{equation}

\noindent where $\xi$ is the standard complex coordinate on $\mathbb{S}^2$ centered at $p^+$ and $J_{\mathbb{S}^2}$ is the standard complex structure of the $2$ sphere. It is straightforward to check that $j_s \to 1$ uniformly as $s \to \infty$ so that the maps $\tilde{f_s}$ are asymptotically holomorphic with respect to standard complex structure compatible with the round metric.

The vortex equation \eqref{s-vortex} on $L \to B(p_j,\epsilon)$ is then pulled back to \eqref{pullbacked s vortex} defined on the left end of the diagram, where the bundle is equipped with background metric $H_s^* := R_j(s)^* H$. The extra factor of $\frac{1}{t_j(s)^2}$ comes from the fact that $\Lambda_s^*$ is taken with respect to the metric rescaled by $t_j(s)$. The section terms are of course invariant throughout these parametrizations. On the other hand, we may construct solutions $[D_s^\dagger,\phi_s^\dagger]$ to \eqref{pullbacked s vortex} directly from holomorphic maps $\tilde{f}_s: B_s \to \pk$ in the manner of $\Phi_s$ in section 3. Namely, we start with background metric $H_s^\dagger := \tilde{f}_s^* H_{FS}$ on $L_s^*$ and the same holomorphic sections $\phi_{s,i}^* := \tilde{f}_s z_i$ on $L_s^*$. Turning $H_s^\dagger$ into the special metric $G_s^\dagger$ via a gauge, we obtain the solutions $[D_s^\dagger,\phi_s^\dagger]$. The two solutions, $[D_s^*,\phi_s^*]$ and $[D_s^\dagger,\phi_s^\dagger]$ correspond to two special metrics, $G_s^*$ and $G_s^\dagger$, that are gauge transformed from $H_s^*$ and $H_s^\dagger$, so that the holomorphic pair $(\tilde{f}_s^* \bar{\p}, \tilde{f}_s^* \phi_s)$ satisfies the second equation of \eqref{pullbacked s vortex}. $G_s^*$ and $G_s^\dagger$ are therefore determined, up to unitary gauge, by the unique solution to one Kazdan-Warner equation, and we have $[D_s^*,\phi_s^*] = [D_s^\dagger,\phi_s^\dagger]$.

We have identified pullbacked vortices $[D_s^*,\phi_s^*]$ with holomorphic maps $\tilde{f}_s$, we wish to study the convergent behavior of vortices from the Gromov compactness of maps. As we have expected, the existences of bubble point for $\tilde{f}_s$ are determined by the relative rates of convergence of $p_{ij}(s)$ to $p_j$, and they are observed through the scale $t_j(s)$. We here observe that bubble point exists when $t_j(s)$ grows proportionally to $s$, in which the scaling gives rise to an {\em affine} vortex.

\vspace{0.2cm}
\noindent {\em $\bullet$ Case 1 (Fast Blow Up) $\frac{s}{t_j(s)} \to 0$ as $s\to\infty$}.

This case has been ruled out by arguments from \cite{GS} due to the asphericality of $\mathbb{C}^{k+1}$.

\vspace{0.2cm}
\noindent {\em $\bullet$ Case 2 (Slow Blow Up): }$\frac{s}{t_j(s)} \to \infty$ as $s\to\infty$.

For this extreme, it follows that $\sum_{i=0}^k |\phi_{s,i}^*|^2_H \to 1$ as $s\to\infty$. In another words, for $s$ large enough, we may assume that

\[\cap_{i=0}^k \phi_{s,i}^{*\;-1} (\{0\}) = \emptyset\]

\noindent and that the limiting sections have no common zero. From \eqref{energy density in coordinate}, we see that the holomorphic maps $\{\tilde{f}_s\}$ have uniformly bounded energy densities and do not have bubble point. Let $\tilde{f}_{p_j}$ be its subsequential $C^1$ limit. The map defines a vortex $[D_j,\phi_j]$, which satisfies \eqref{pullbacked infinity vortex holomorphic sphere}. Indeed, since $e(\tilde{f}_s^*) $, and therefore $\sqrt{-1}\Lambda_s^* F_{D_s^*}$ in \eqref{pullbacked s vortex}, are uniformly bounded, the curvature term there approaches $0$ after dividing by $\frac{s^2}{2(t_j(s))^2}$ and let $s \to \infty$. The holomorphic condition clearly continues to hold as $s \to \infty$ and \eqref{pullbacked infinity vortex holomorphic sphere} follows. We therefore obtain a finite energy holomorphic map from $\mathbb{C}$ to $\pk$, which extends to a map from $\mathbb{S}^2$ to $\pk$, or a holomorphic sphere in $\pk$.

\vspace{0.2cm}
\noindent {\em $\bullet$ Case 3 (Moderate Blow Up): }$\frac{s}{t_j(s)} \to \lambda \in (0,\infty)$ as $s\to\infty$.

Bubble points occur only in this case, when energy densities blow up roughly proportional to $s$. With an additional normalization if necessary, we assume that $\lambda=1$. By conformal invariance of energy, we have $E(\tilde{f}_s) = a_j$ for all $s$. Following constructions in the proof of Theorem \ref{renormalization}, let $\{p_j^1,\ldots,p_j^{N_j},p^+\}$ be the bubble points, $\tilde{f}_j$ be the $C^1_{loc}$ limit of $\tilde{f}_s$ on $\Sigma \backslash \{p_j^1,\ldots,p_j^{N_j}\}$, and $\bar{f}_j$ be the extension of $\tilde{f}_j$ to $\mathbb{S}^2$. Let $a_j^l \in \mathbb{N}$ be defined as in Theorem \ref{renormalization}:

\begin{equation}
a_j^l := \lim_{\epsilon \to 0}\limsup_{s\to\infty} \int_{B(p_j^l,\epsilon)} \left||e(\tilde{f}_s)|-|e(\bar{f}_j)|\right| \omega_s.
\label{energy loss at p j l}
\end{equation}

\noindent The vortex $[D_j,\phi_j]$ defined by $\tilde{f}_j$ therefore satisfies \eqref{pullbacked infinity vortex}, where $a_j^0 = E(\bar{f}_j)$.

We finally check that the limiting procedure is compatible with the geometric structure of line bundles $L_s^*:=\tilde{f}_s^* \mathcal{O}(1)$ over $B_s$. The degrees of $L_s^*$ are precisely the energies of $\tilde{f}_s$, which is the constant $a_j$. Each $L_s^*$ is equipped with transition functions $\{\gamma_{\alpha \beta}^s\}$ (pullbacked from those on $\mathcal{O}(1)$ via $\tilde{f}_s$) so that on the overlap, we have the compatibility condition

\begin{equation}
A_{s,\alpha} = -d\gamma_{s,\alpha \beta} \left(\gamma_{s,\alpha\beta}\right)^{-1} + \gamma_{s,\alpha\beta} A_{s,\beta}\left(\gamma_{s,\alpha\beta}\right)^{-1},
\label{family of transitions}
\end{equation}

\noindent where $A_{s,\alpha}$ and $A_{s,\beta}$ are the connection forms of the unitary connections $D_s^*$ constructed above. That is, the coefficients of these one forms are by definition algebraic expressions $\tilde{f_s}$ and its first derivatives. For a local holomorphic frame $u_\alpha$, we have

\begin{equation}
A_{s,\alpha} = d^\prime \log H_{s,\alpha}^*  = d^\prime \log \left(\frac{|u_\alpha|^2}{\sum_{i=0}^k |\tilde{f}_s (z_i)|^2}\right).
\label{local expression of connection one forms}
\end{equation}

\noindent The $C^1_{loc}$ convergence of $\tilde{f}_s$ therefore imply that $A_\alpha^s$ and $A_\beta^s$ converge uniformly (in $s$) on their domains of definitions. Smooth convergence of these transition functions then follow by standard bootstrapping arguments. As $s \to \infty$, \eqref{family of transitions} passes to the limit:

\begin{equation}
A_{j,\alpha} -d\gamma_{j,\alpha \beta} \left(\gamma_{j,\alpha\beta}\right)^{-1} + \gamma_{j,\alpha\beta} A_\beta\left(\gamma_{j,\alpha\beta}\right)^{-1}
\label{family of transitions in the limit}
\end{equation}

\noindent on $\mathbb{S}^2 \backslash \{p_j^1,\ldots,p_j^{N_j},p^+\}$. Therefore, the local transition functions and one forms patch together to give a degree $a_j^0$ holomorphic line bundle $L_j$ over the punctured sphere. This proves the second statement.

\noindent Finally, since the conformal transformations are designed so that there is no energy loss at $p^+$, and therefore

\begin{equation}
  E(\tilde{f}_j) = a_j - \sum_{l=1}^{N_j} a_j^l := a_j^0.
  \label{sum of energy conformal}
\end{equation}

\noindent The curvature current $F_{D_j}$ it defines then satisfies the last statement.

\end{proof}

Iterating the procedure at each bubble point $p_j^l \in \mathbb{S}^2$, the theorem provides an almost complete description on the root of bubble tree. We however still face the hurdle of extending the vortex across $p^+$, or removal of singularities. Extending holomorphic maps of finite energy over $p^+$ poses little difficulty by application of Theorem \ref{removable singularity for maps}, but the extensions are not necessarily compatible with the correspondence of maps and vortices. This is the stage where classical extension results for topological fields, such as those in \cite{U1} and \cite{U2}, enter. The more general case of critical points to the Yang-Mills-Higgs energy functional, which contains our case, have been discussed in \cite{S}. We briefly summarize the work here.

\section{Extension of the bundle at infinity and conical metrics}
\subsection{The smooth case}
Fix a small geodesic disc $B_{R_0}$ around $p+$ within the injectivity radius that contains no other bubble point. The removability of singularities depends on how well one controls certain norms of connection form $A$ and curvature $F_A$ on $B_{R_0} \backslash \{p^+\}$ (cf. \cite{U1},\cite{U2} for Yang-Mills fields). For critical points to $YMH_s$, Smith showed that the limiting field may be extended across $p^+$ if $F_A$ is integrable enough and there is a gauge on which

\begin{equation}
  \left\Vert A \right\Vert_p |_{B_{R_0}\backslash \{p^+\}} \leq O(\rho),
  \label{decaying condition for A smooth}
\end{equation}

\noindent where $\rho=|\xi|$ is the radial component of the polar coordinate of $B_{R_0}$.  Such an estimate allows one to apply implicit function theorem to solve the required regularity condition

\begin{equation}
  d^*(g^{-1}dg + g^{-1}A g)=0,
  \label{regularity condition}
\end{equation}

\noindent which ensures the smoothness of extended field on the entire disc $B_{R_0}$. The establishment of \eqref{decaying condition for A smooth} has been thoroughly discussed in \cite{S}. It was proved that the bound follows from certain decay condition on the holonomy of the connection $D$, called the "Condition H":

\begin{definition}[Condition $H$]
Let $D$ be a an affine connection on the bundle $L$ over $B_{R_0}$ and $\gamma_R(t): [0,1] \to \p B_{R_0}$ be a smooth positive parametrization of the circle $\p B_{R_0}$ around $p^+$. Let $g(R)$ be the holonomy of $D$ over the $\gamma (t)$. That is, for every $D$-parallel vector field $v$ over $\gamma(t)$, we have

  \[v(\gamma(1))=v(\gamma(0))\cdot g(R).\]

\noindent We say that $D$ satisfies {\em connection $H$} if

\begin{equation}
  \lim_{R \to 0} g(R) = id,
  \label{condition H}
\end{equation}

\noindent pointwise.
\end{definition}

This condition is in fact equivalent to the existence of the gauge, on which the angular component of the connection form $A$ decays to $0$.

\begin{theorem}[Theorem 1.1 in \cite{S}]
  Condition $H$ is equivalent to the existence of a unitary gauge in which

  \[A=A_\rho (\rho,\theta) d\rho + A_\theta (\rho,\theta) d\theta \in \mathfrak{u}(1)\]

  \noindent with

  \begin{equation}
    \lim_{\rho\to 0} A_\theta (\rho,\theta) = 0
  \label{decay of A radial smooth}
  \end{equation}
  \noindent in sup norm topology.
\end{theorem}

\noindent With \eqref{decay of A radial smooth}, one may apply further gauge transformation, or the "auxiliary gauge" to improve the decay of $A$.

\begin{theorem}[The Auxiliary Gauge, cf. Lemma 1.1 of \cite{S}]
  If the connection $D$ satisfies the $H$ condition \eqref{decay of A radial smooth}, there exists a gauge in which $D=d+A$ and the followings hold:

  \begin{equation}
    \lim_{\rho \to 0} A_\rho(\rho,0)=0,\hspace{1cm}\lim_{\rho\to 0}A_\theta (\rho,\theta)=0,\hspace{1cm} \lim_{\rho\to 0} \frac{\p}{\p \rho}A_\theta (\rho,\theta)=0.
    \label{auxiliary gauge}
  \end{equation}
\end{theorem}

\noindent In the auxiliary gauge, one can estimate the $L^p$ norms of $A_\rho$ and $A_\theta$ separately, as in the technical Lemmas 1.2 and 1.3 in \cite{S}, so that \eqref{decaying condition for A smooth} follows. Classical arguments in \cite{U1} and \cite{U2} imply the theorem on removability of singularity:

\begin{theorem}[Theorem M in \cite{S}, Relevant Form]\label{theoremM}
On $B_{R_0} \backslash \{p+\}$ with Euclidean metric and $L$ a line bundle over it, let $A$ be a connection form satisfying condition $H$, and its curvature form $F \in L^p(B_{R_0})$ be smooth for $p\geq 1$. Assume that $(F,\phi)$ satisfies the Euler-Lagrange equation of Yang-Mills-Higgs energy functional \ref{YMH}, and $\phi \in H_2^1(B_{R_0})$. Then, there exists a continuous gauge transformation such that $(F,A)$ is gauge equivalent to a smooth pair $(\tilde{F},\tilde{\phi})$ over $B_4^2$ and the bundle extends smoothly to $B_{R_0}$.
\label{Theorem M}
\end{theorem}

\subsection{Conical Hermitian metrics}

Unfortunately, the hypotheses of theorem \ref{theoremM} are not guaranteed for the limiting vortex $[D_j,\phi_j]$ constructed in Theorem \ref{renormalization}. The first source of failure is the integrability of curvature form $F_j$ of $D_j$. In the notations of diagram \eqref{pull back diagram}, let again $z$ and $\xi$ be natural complex coordinates of $B(p_j,\epsilon) \subset \Sigma$ and $B_s \subset \mathbb{S}^2$, centered at $p_j$ and $p^+$, respectively. The K\"ahler form $\omega$ on $B(p_j,\epsilon)$ can be expressed as

\begin{equation}
  \omega(z) = \frac{1}{2\pi \sqrt{-1}} g(z) dz \wedge d\bar{z},
  \label{Kahler form on Sigma}
\end{equation}

\noindent with $g(0) = 1$. With further holomorphic coordinate change if necessary, the pullback K\"ahler form is

\begin{equation}
\omega_s^*(\xi) := R_j(s)^* \omega(\xi) = \frac{1}{2\pi \sqrt{-1}} g \left(\frac{\xi}{t_j(s)}\right) \frac{1}{(1+|\xi|^2)^2} d \xi \wedge d \bar{\xi},
\label{Kahler form pull back}
\end{equation}

\noindent and it is clear that $\omega_s^*$ approaches the standard round metric as $s \to \infty$. In other words, the limiting vortex the restriction of an affine vortex on $B_{R_0} \backslash \{p^+\} \simeq \mathbb{C} \backslash B(0,R_0^\prime)$ for some large $R_0^\prime$. For such a vortex, the pointwise norm of $F_\infty^*$ is controlled by the estimate of energy density given in \cite{Z}:

\begin{proposition}[\cite{Z}, Corollary 1.4]\label{z}
Let $\omega = dx\wedge dy$ be the standard area form on $\mathbb C$ and assume that $(A,\phi)$ is an affine vortex with such that the images of sections have compact closure. Define the energy density by

\[e(A,\phi) := |F_A|^2 + \sum_{i=0}^k |D_A \phi_i|^2 + \frac{1}{4} \left|\sum_{i=0}^k |\phi_i|_H^2-1\right|^2.\]

\noindent Then for every $\epsilon >0$ there exists a constant $C_\epsilon$ such that:

\begin{equation}
e(A,\phi) \leq C_\epsilon \, |z|^{-4+\epsilon}
\label{energy density bound affine vortex}
\end{equation}

\noindent for $|z|\geq 1$.
\end{proposition}

\noindent Pulling back estimate \eqref{energy density bound affine vortex} onto $\mathbb{S}^2 \backslash \{p^+\} \simeq \mathbb{C}$ as above, the estimate for pointwise norm of $F_j$ we have is

\begin{equation}
  |F_j| \leq C_\epsilon^\prime |\xi|^{-2 - \frac{\epsilon}{2}}
  \label{curvature estimate with round metric}
\end{equation}

\noindent for some $\epsilon >0$. The estimate is clearly insufficient to guarantee any integrability of $F_j$ in the usual Lebesgue measure.

For the vortex $[D_j,\phi_j]$ that does not satisfy the H condition \eqref{decay of A radial smooth}, a reasonable modification is to associate a {\em conic} singularity at $p^+$ to absorb the singularity. We recall the following definition of canical metrics in dimension ~2:
\begin{definition}\label{conicmetric}
Let $\Sigma$ be a Riemann surface and $p\in \Sigma$. A conical K\"ahler metric of angle $\beta$, conical at $p$, is a metric whose K\"ahler $(1,1)$-form in a holomorphic coordinate system centered at $p$ looks like:
$$\omega = e^u \frac{dz\wedge d\bar z}{|z|^{2-2\beta}}$$
\noindent for some $\beta \in (0,1)$ near $p$. Here $u\in C^0(\Sigma)$.
\end{definition}
Such metrics can be realized as the pullback metric of the map $z \to z ^\beta$ and also the map $w\to w^{\frac{1}{\beta}}$, where defined, pulls such metrics back to a smooth metric.
\noindent
It is an easy calculation to see that the form of a conic metric at $\infty$ looks like:
\begin{equation}
  \omega_{\infty,\beta}^* :=e^u \;\frac{|\xi|^{-2+2\beta}}{2\pi \sqrt{-1}} \frac{1}{(1+|\xi|^2)^2} d\xi \wedge d\bar{\xi}
  \label{conic metric local expression}
\end{equation}
where $\xi= \frac{1}{z}$.

We now recall some properties of the function spaces described by Donaldson \cite{D} and their behavior under the Laplacian (cf. same reference).
Introducing a background metric $g_\beta$, conical along a divisor $D$, with associated distance $d_\beta$, let:
$$ \mathcal{C}^{,\alpha,\beta}(M,D) = \left\{ f \in \mathcal C^0(M) \mbox{ s.t. } \| f \|_{,\alpha,d_\beta} < +\infty \right\},$$
$$ \mathcal{C}^{,\alpha,\beta}_0(M,D) = \left\{ f \in \mathcal C^{,\alpha,\beta}(M,D) \mbox{ s.t. } f(x)=0 \mbox{ for all } x \in D \right\},$$
where
$$ \| f \|_{,\alpha,d_\beta} = \sup_{x\in M} |f(x)| + \sup_{x,y \in M} \frac{|f(x)-f(y)|}{d_\beta(x,y)^\alpha}.$$

As the notation suggests, those spaces depend on $\beta$ but they are independent of the particular conical metric $g_\beta$ chosen.
This follows from the fact that any two metrics on $M$ which are conical of angle $\beta$ along $D$ induce equivalent distances on $M$. We will sometimes refer to functions in these spaces as {\em{$\beta$-weighted}} functions of a given regularity.

In local holomorphic coordinates $z$ centered at a point of $D$ such that $D$ is the locus $z=0$ define
$$\Phi (z) =  |z|^{\beta -1} z$$ which is clearly a homeomorphism.
As noted by Donaldson \cite{D}, a function $f$ on $M$ is of class $\mathcal C^{,\alpha,\beta}$ if and only if it is $C^{0,\alpha}$ away from $D$, and $f \circ \Phi^{-1}$ is $C^{0,\alpha}$ for any choice as above of local coordinates around $D$.

Consider the change of coordinates $z=\psi(w)$, where the map $\psi$ is defined by
$$ \psi(w)=w^\frac{1}{\beta}, $$
with $w = \rho e^{\sqrt{-1}\theta}$, $0<\theta<\frac{2\pi\beta}{1+\beta}$.
Note $\psi$ is a biholomorphism on its image, therefore $(w$ is local holomorphic coordinates around any point not lying in $D$.
One can check that $\Phi \circ \psi$ is bi-Lipschiz, whence it follows that $\mathcal C^{,\alpha,\beta}$ is constituted by functions on $M$ that are $C^{0,\alpha}$ away from $D$, and such that $\psi^*f$ is $C^{0,\alpha}$ for any choice of local (holomorphic) coordinate $z$ around $D$ as above.

Now we pass to recall the definition of one and two forms of class $\mathcal C^{,\alpha,\beta}$. A $(1,0)$-form $\xi$ on $M$ is said to be of class $\mathcal C^{,\alpha,\beta}$ if it is $C^{0,\alpha}$ away from $D$, and $\psi^*\xi$ is of class $C^{0,\alpha}$ and it satisfies $\psi^*\xi\left(\frac{\partial }{\partial w}\right) \to 0$ as $w \to 0$. Analogously, a $(1,1)$-form $\eta$ on $M$ is said to be of class $\mathcal{C}^{,\alpha,\beta}$ if it is $C^{0,\alpha}$ away from $D$ and $\psi^*\eta$ is $C^{0,\alpha}$, and both the contractions of $\psi^*\eta$ with $\frac{\p}{\p w}$ or $\frac{\p}{\p \bar w}$ go to zero as $w \to 0$. One then defines:

$$\mathcal C^{2,\alpha,\beta} = \left\{f \in C^2(M \setminus D) \cap C^0(M) \mbox{ s.t. } f,\partial f,\partial\bar{\partial} f \mbox{ are of class } \mathcal C ^{,\alpha,\beta} \right\}.$$

\noindent We also consider the space:
$$\mathcal H:= \left\{f \in \mathcal{C}^{\infty}(M)  \text{ where } \Phi^*f  \in W^{1,2}\right\}$$

\noindent Then Donaldson proves the following (cf. \cite{adl} for a strengthening) :
\begin{theorem}[S.~Donaldson \cite{D}]\label{donaldson-regularity}
Let $\omega$ be a K\"ahler metric on the ball $B_6(0)$ which is of class $\mathcal{C}^{,\alpha,\beta}$ and satisfies $a_1 \, \Omega \leq \omega \leq a_2 \, \Omega$ for some suitable constants $a_1,a_2>0$. Suppose that $\alpha < \mu:= \frac{1}{\beta}-1$ and that $f$ is a function of class $\mathcal{C}^{,\alpha,\beta}$ defined on the ball $B_6(0)$, and $v \in \mathcal{H}$ is a weak solution of the equation $\Delta_\omega v = f$. Then the restriction $v|_{B_1(0)}$ is of class $\mathcal{C}^{2,\alpha,\beta}$.
\end{theorem}


  With $\beta > \frac{\epsilon}{2}$, we see that $F_j \in L^1_\beta (B_{R_0})$, the space of integrable functions with measure defined by $\omega_{\infty,\beta}^*$ above. Following identical arguments in the proof of Theorem 4.1 in \cite{S} (with the condition that $\phi_j$ is smooth), we conclude that $F_j \in L^p_\beta (B_{R_0})$ for all $p \geq 1$.

Next, we naturally generalize the condition H above with conic metric. A Hermitian metic $h$ on a line bundle is conic if $h(s)= |z|^{2\beta}$ for holomorphic frame $s$.
The associated connection is
$$D= d+ i\beta d\theta$$
or, in holomorphic coordinates:
$$d+\beta \frac{dz}{z}$$
We view this as the standard model.

 For $\beta \in (0,1)$, an extra factor of $\rho^{2\beta-2}$ appears in the volume measure and to achieve the decay condition

\begin{equation}
  \left\Vert A \right\Vert_{p,\beta} |(B_{R_0}\backslash \{p^+\}) \leq O(\rho),
\end{equation}

\noindent the corresponding "$H_\beta$ condition" for integral estimates to hold true is then

\begin{equation}
   \lim_{\rho\to 0} \rho^{2\beta-2} A_\theta (\rho,\theta) = 0.
   \label{H beta condition}
\end{equation}

\noindent The analogous theorem for removal of singularity is then the following, whose proof is entirely identical to Theorem \ref{Theorem M} except the measure on the $L^p$ space is replaced by the one defined by conic metric. However, with \eqref{H beta condition}, all the estimates in \cite{S} remain valid and making substantial use of Theorem \ref{donaldson-regularity} we have

\begin{theorem}[Theorem $M_\beta$]
On $B_{R_0} \backslash \{p+\}$ with conic metric of angle $2\pi\beta$ at $p^+$, $L$ a line bundle over it, let $A$ be a connection form satisfying condition $H_\beta$, and its curvature form $F \in L^1_\beta(B_{R_0})$ be smooth. Assume that $(F,\phi)$ satisfies the Euler-Lagrange equation of Yang-Mills-Higgs energy functional \ref{YMH}, and $\phi \in \mathcal H (B_{R_0})$. Then, there exists a continuous gauge transformation such that $(F,A)$ is gauge equivalent to a smooth conic pair $(\tilde{F},\tilde{\phi})$ over $B_{R_0}$, the bundle extends smoothly to $B_{R_0}$ and:
$$\tilde A=d+ i\beta d\theta + \mathfrak a$$
where $\mathfrak a$ is smooth.
\label{Theorem M beta}
\end{theorem}

For the vortex $[D_j,\phi_j]$ in particular, it remains to construct a trivialization in which the connection form of the limiting vortex $(D_j,\phi_j)$ on $B_{R_0} \backslash \{p^+\}$ satisfies condition $H_\beta$.

\begin{theorem}
  For the line bundle $L_j$ and the limiting gauge $[D_j,\phi_j]$ in the second statement of Theorem \ref{renormalization} failing to satisfy the H condition \eqref{decay of A radial smooth}, we may continuously extended it to $(\mathbb{S}^2,g_\beta)$, where $g_\beta$ is the conic metric of angle $\beta$ at the north pole.
  \label{removal of singularity on sphere}
\end{theorem}

\begin{proof}
Let $A_j = A_j(z) dz$, where $A_j(z) \in \mathfrak{u}(1)$. On the domain $\mathbb{C} \backslash B(0,R_0^\prime)$, consider, for some $\alpha > 3-2\beta>1$, the differential equation

\begin{equation}
  A^\gamma_j=\frac{\p}{\p z} \log \gamma + A_j(z) = \frac{\sqrt{-1}}{|z|^\alpha}dz.
  \label{ODE for decaying gauge}
\end{equation}

The gauge $\gamma$ can be easily expressed by integration: (Note that the right hand side is integrable since $\alpha >1$.)

\begin{equation}
  \gamma = \exp\left(\sqrt{-1}\int_{z_0}^z \left(\frac{1}{|z|^\alpha}-A_j(z')\right)dz'\right) \in U(1).
  \label{decaying gauge}
\end{equation}

\noindent In polar coordinate $(r,\theta)$ of $B_{R'}$, $A_\gamma$ can be written as

\begin{equation}
A^\gamma_j = e^{\sqrt{-1}\theta}\left(\frac{\sqrt{-1}}{r^\alpha} dr - \frac{1}{r^{\alpha-1}} d\theta\right).
\label{decaying connection form polar}
\end{equation}

\noindent Pulling back to $\mathbb{S}^2$ via $R_j(s)$ and let $s\to\infty$, we have, in polar coordinate $(\rho,\theta)$ of $\mathbb{S}^2$ near $p^+$, that

\begin{equation}
  (A^\gamma_j)^*_\theta:=R_j(\infty)^* (A^\gamma_j)_\theta = \rho^{\alpha-1} G(\rho,\theta),
  \label{decaying connection form polar on sphere}
\end{equation}

\noindent where $G$ is a smooth nonvanishing function on $B_{R_0}$. Since $\alpha-1-(2-2\beta)>0$, \eqref{H beta condition} holds and we are done.

\end{proof}
\subsection{Bubble trees with cones}

The description of bubble tree at the root level is now complete, and the entire bubble tree is essentially a finite number of iterations of these renormalization processes. To unify the notation, we denote $\Sigma$ by $T_0$ and relabel the bubble points by

\begin{equation}
  \mathbb{B}^0:= \{p_1^0,\ldots,p_{N_0}^0,q_1^0,\ldots,q_{N_0^\prime}^0,r_1^0,\ldots,r_{N_0^{\prime\prime}}^0\}\subset T_0,
  \label{bubble point at T 0}
\end{equation}

\noindent where $p$,$q$, and $r$ are bubble points where round spheres, conic spheres (raindrops), and holomorphic spheres that are bubbled off respectively. The first level $T_1$, is then a disjoint union of these $N_0+N_0^\prime + N_0^{\prime\prime}$ bubbles with various types:

\begin{equation}
  T_1 := \left(\bigsqcup_{i=1}^{N_0} \mathbb{S}^2_{p_i^0}\right) \sqcup \left(\bigsqcup_{j=1}^{N_0^\prime} \mathbb{S}^2_{q_j^0}\right) \sqcup \left(\bigsqcup_{l=1}^{N_0^{\prime\prime}} \mathbb{S}^2_{r_l^0}\right),
  \label{bubble tree first level}
\end{equation}

\noindent  Each sphere is wedged to its designated bubble point at its north pole, and those spheres from the first two components above may contain new bubble points. We similarly classify them by the types of new bubbles they form, as in $\mathbb{B}^0$:

\begin{equation}
  \mathbb{B}^1:=\{p_1^1,\ldots,p_{N_1}^1,q_1^1,\ldots,q_{N_1^\prime}^1,r_1^1,\ldots,r_{N_1^{\prime\prime}}^1\}\subset T_1.
  \label{bubble point at T 1}
\end{equation}

\noindent $\mathbb{B}^1$ give rise to a new set of $N_1+N_1^\prime +N_1^{\prime\prime}$ bubbles whose disjoint union is $T_2$ with new bubble points $\mathbb{B}_2$. The process is iterated and we have the main theorem of this article:

\begin{theorem}[Bubble Tree]
The vortices $V_s:=\{[D_s,\phi_s]\}$ on a degree $r$ line bundle $L$ over $\Sigma$ Gromov converge to a vortex $\mathcal{V}:=[\mathcal{D},\varPhi]$ over a degree $r$ line bundle $\mathcal{L}$ over a bubble tree $\mathcal{T}$ defined by

\begin{equation}
  \mathcal{T}:= T_0 \vee T_1 \vee \cdots  \vee T_{N_{\mathcal{V}}},
  \label{bubble tree wedges}
\end{equation}

\noindent where $T_0 = \Sigma$.

\end{theorem}

\begin{proof}

$T_n$ is constructed inductively as above, with associated bubble points $\mathbb{B}_n$. $T_n$ and $T_{n+1}$ are then wedged at $\mathbb{B}_n \subset T_n$ and the north poles of each sphere in $T_{n+1}$ designated to its bubble point. By Lemma \ref{ghost bubble} and  Theorem \ref{minimum energy requirement}, each renormalization reduces the total energy by at least $C_0 > 0$ and therefore the bubble tower consists of at most finite number ($N_\mathcal{V}$) of levels.

$\mathcal{L}$ is the line bundle over $\mathcal{T}$ whose restriction to each sphere $\mathbb{S}^2_{p_i^n} \backslash \{p^+\}$ is the holomorphic line bundle determined by the connection $D_{p_i^n}$ that is determined by iterations of Theorem \ref{renormalization}. $\varPhi$ is constructed identically. Allowing the possibility of conic singularity, the discussions above conclude that the extensions of vortices are possible whenever the corresponding holomorphic maps extend. Therefore, we may define prolongation of vortices $V_s$, $\mathcal{V}_s$, identically by their corresponding maps. The Gromov convergence of $V_s$ to $\mathcal{V}$ are then defined by the $C^1$ convergence from $\mathcal{V}_s$ to $\mathcal{V}$ on $\mathcal{T}$, which follow from convergence of their corresponding maps. Conic singularities are appropriately introduced at points where energy density does not decay fast enough. Since energy is conserved throughout the entire process, the degree of $\mathcal{L}$, or the sum of the energies of $\mathcal{V}$ on $\Sigma$ and all bubbles, is precisely the original degree $r$.

\end{proof}

The Gromov limit of a sequence of vortices have been constructed. The natural follow up construction is the moduli space whose boundary includes all these bubble trees. Furthermore, we expect some kind of dynamics, or $L^2$ metric on the space. The rate of convergence, or equivalently the rate of coalescence of zeros, will play an important role in this metric. We are eager to purse, or learn any possible progress in this direction.

\section{Acknowledgement}

The second author is supported by grant 105-2115-M-006-012, Ministry of Science and Technology of Taiwan.


\begin{thebibliography} {Library}

\bibitem [ADL]{adl} C. Arezzo], A. Della Vedova, G. La Nave
{\em On the curvature of conic K\"ahler - Einstein metrics}, arXiv:1608.06890  to appear on Journal of Geometric Analysis

\bibitem [B]{Ba}
J.M. Baptista,
{\em On the $L^2$ Metrics of Vortex Moduli Spaces}, Nuclear Physics B, 844, 308-333 (2010).

\bibitem [B1]{Ba1}
J.M. Baptista,
{\em Moduli Spaces of Abelian Vortices on K\"ahler Manifolds}, arXiv: 1211.012.


\bibitem [Br] {Br}
S.B. Bradlow,
{\em Vortices in Holomorphic Line Bundles over Closed K\"ahler Manifolds}, Commun. Math. Phys. 135, 1-17 (1990).

\bibitem [Br1] {Br1}
S.B. Bradlow, {\em Special Metrics and Stability for Holomorphic Bundles with Global Sections}, J. Diff. Geom. 33, 169-213 (1991).

\bibitem [B-D-W]{BDW}
A. Bertram, G. Daskalopoulos, and R. Wentworth,
{\em Gromov Invariants for Holomorphis Maps from Riemann Surfaces to Grassmannians}, Journal of the American Mathematical Society. 9, 529-571 (1996).

\bibitem [D] {D}
S.K. Donaldson,
{\em K\"ahler Metrics with Cone Singularities Along a Divisor}, Essays in Mathematics and its Applications, 49-79, (2012).


\bibitem [Gri] {Gri}
P.A. Griffiths, {\em Introduction to Algebraic Curves}, American Mathematical Society, Vol 76, (1983).

\bibitem [Gro] {Gro}
M. Gromov,
{\em Pseudo Holomorphic Curves in Symplectic Manifolds}, Invent. Math. 82, 307-347 (1985).


\bibitem [GS] {GS}
A. R. Gaio, D.A. Salamon,
{\em Gromov-Witten Invariants of Symplectic Quotients and Adiabatic Limits}, J. Sympl. Geom., 3(1), 55-159, (2005).

\bibitem [J-T] {JT}
A. Jaffe, C. Taubes,
{\em Vortices and Monopoles}, Birkh\"auser, (1981).

\bibitem [K]{K}
S. Kobayashi, {\em Differential Geometry of Complex Vector Bundles}, Iwanami Shoten, Publishes and Princeton University Press, 1987.


\bibitem [K-W]{kw}
J. Kazdan, F.W. Warner,
{\em Curvature Functions for Compact ~2-Manifolds}, Ann. Math 2, 99, 14-47 (1978).

\bibitem [L]{L}
C. Liu,
{\em Dynamics of Abelian Vortices Without Common Zeros in the Adiabatic Limit}, Commun. Math. Phys., 329, 169-206 (2014).

\bibitem [Ma] {Ma}
N.S. Manton,
{\em A Remark on the Scattering of BPS Monopoles}, Phys. Lett. 110B, 54-56 (1982).


\bibitem [Mo] {Mo}
C.B. Morrey,
{\em Multiple Integrals in the Calculus of Variations}, New York: Springer-Verlag, 1966.






\bibitem [PW] {PW}
T.H. Parker, J.G. Wolfson,
{\em Pseudo-Holomorphic Maps and Bubble Trees}, J. Geom. Phys, Vol. 3, No. 1, 63-98 (1993).




\bibitem [S] {S}
P.D. Smith
{\em Removable Singularities for the Yang-Mills-Higgs Equations in Two Dimensions}, Annales de l'I.H.P. Analyse non lin\'eaire, Vol. 7, Iss. 6, 561-588 (1990).

\bibitem [S-U] {SU}
J. Sacks, K. Uhlenbeck,
{\em The Existence of Minimal Immsersions of 2-Spheres}, Ann. Math. 113, 1-24 (1981).


\bibitem [U1]{U1}
K.K. Uhlenbeck,
{\em Removable Singularities in Yang-Mills Fields}, Commun. Math. Phys., 83, 11-29, (1982).

\bibitem [U2]{U2}
K.K. Uhlenbeck,
{\em Connections with $L^p$ Bounds on Curvature}, Commun. Math. Phys., 83, 31-42, (1982).

\bibitem [W] {W}
J.G. Wolfston,
{\em Gromov's Compactness of Pseudo-Holomorphic Curves and Symplectic Geometry}, J. Diff. Geom., 28, 383-405 (1988).

\bibitem [X] {X}
G. Xu,
{\em Classification of $U(1)$-Vortices with Target $\mathbb{C}^N$}, Int. J. Math., Vol. 26, No. 13 (2015).
\bibitem [Z]{Z}
F. Ziltener,
{\em The invariant symplectic action and decay for vortices}, J. Symplectic Geom. Volume 7, Number 3 (2009), 357-376.



\end{thebibliography}
\end{document}